\documentclass[10pt, compsoc, journal]{IEEEtran}
\usepackage{cite}
\usepackage{graphicx}
\usepackage{lipsum}
\usepackage[cmex10]{amsmath}
\usepackage{multirow}
\usepackage{amsthm}
\usepackage{float}
\usepackage{color}
\usepackage[lined,linesnumbered, ruled]{algorithm2e}
\usepackage{caption}
\usepackage{booktabs}
\usepackage{ragged2e}

\usepackage{array}

\usepackage{subcaption}

\usepackage{enumitem}

\newtheorem{theorem}{Theorem}
\newtheorem{lemma}{Lemma}

\DeclareCaptionLabelSeparator{periodspace}{.\quad}
\begin{document}
\newcommand{\notehaibo}[1]{\textcolor{blue}{[haibo: #1]}}

\title{Accelerating Fully Connected Neural Network on Optical Network-on-Chip (ONoC)}
\vspace{-2mm}
\author{Fei Dai, Yawen Chen, Haibo Zhang and Zhiyi Huang
\thanks{F. Dai, Y. Chen, Z. Huang, and H. Zhang are with the Department of Computer Science, University of Otago, Dunedin 9016, New Zealand, e-mail: (travis, yawen, haibo, zhuang@cs.otago.ac.nz).}
\thanks{Manuscript received XX XX, 20XX; revised XX XX, 20XX.}
}
\IEEEtitleabstractindextext{
\begin{abstract}
	Fully Connected Neural Network (FCNN) is a class of Artificial Neural Networks widely used in computer science and engineering, whereas the training process can take a long time with large datasets in existing many-core systems. Optical Network-on-Chip (ONoC), an emerging chip-scale optical interconnection technology, has great potential to accelerate the training of FCNN with low transmission delay, low power consumption, and high throughput. However, existing methods based on Electrical Network-on-Chip (ENoC) cannot fit in ONoC because of the unique properties of ONoC. In this paper, we propose a fine-grained parallel computing model for accelerating FCNN training on ONoC and derive the optimal number of cores for each execution stage with the objective of minimizing the total amount of time to complete one epoch of FCNN training. To allocate the optimal number of cores for each execution stage, we present three mapping strategies and compare their advantages and disadvantages in terms of hotspot level, memory requirement, and state transitions. 
	Simulation results show that the average prediction error for the optimal number of cores in NN benchmarks is within 2.3\%. 
	We further carry out extensive simulations which demonstrate that FCNN training time can be reduced by 22.28\% and 4.91\% on average using our proposed scheme, compared with traditional parallel computing methods that either allocate a fixed number of cores or allocate as many cores as possible, respectively. Compared with ENoC, simulation results show that under batch sizes of 64 and 128, on average ONoC can achieve 21.02\% and 12.95\% on reducing training time with 47.85\% and 39.27\% on saving energy, respectively.

\end{abstract}

}

\maketitle
\IEEEdisplaynontitleabstractindextext
\IEEEpeerreviewmaketitle

\vspace{-1mm}
\section{Introduction}
\noindent Artificial Neural Networks (ANNs), such as Fully Connected Neural Network (FCNN), Convolution Neural Network (CNN) and Recurrent Neural Network (RNN), are very popular nowadays. 
FCNN known as Multilayer perceptron (MLP) has the architecture that all the neurons in one layer are connected to the neurons in the next layer, which is widely employed in the applications related to prediction, pattern classification and function approximation in practical. Theoretically, FCNN is able to approximate any functions with a degree of loss and had already been proven to be a universal function approximator~\cite{chen2019hardware}. It is reported that FCNN dominates the communication workloads in data centers, where FCNN, CNN and RNN represent 61\%, 19\% and 5\% of the total workload respectively~\cite{jouppi2017datacenter}. Moreover, the fully connected network structure is commonly used by CNN to implement classification. However, FCNN training process can take a long training time with large data sets mainly due to the data movement latency in the hierarchical memory architecture~\cite{memory1}. Therefore, it is of great significance to break the bottleneck of training FCNN.  

To accelerate the training of FCNN, parallel computations are adopted by using parallel programming architectures, such as MPI, OpenCL, OpenMP, CUDA, and etc.
With these parallel programming architectures, FCNN training model or training data can be divided and then assigned to different computing units or platforms. Nevertheless, the latency caused by the data movement from the traditional bus-based main memory to cache, is hard to satisfy the demand of high performance computing in many-core processors.
Though Electrical Network-on-Chip (ENoC) has been proposed to replace the traditional communication methods, transmission delay and power consumption are still two major concerns~\cite{optical}. With more and more cores integrated in one single chip, parallel computing for FCNN training can easily reach communication bottleneck due to the limitations of ENoC.

To break the communication bottleneck, ONoC has been proposed as a promising alternative with the recent development of CMOS-compatible optical devices~\cite{Infocom}.
Instead of using electrical signal, ONoC uses optical signals to transmit data through waveguides with obvious advantages over ENoC, including low transmission delay, low power cost, high bandwidth and high throughput. 
According to~\cite{nicolescu2017photonic}, the average transmission latency in ENoC is 10 times more than that in ONoC, and ENoC has much higher power consumption than ONoC. 
Moreover, ONoC enables multiple optical signals to be transmitted simultaneously in one waveguide using different wavelengths by Wavelength Division Multiplexing (WDM)~\cite{Shacham1}. 
With these advantages, ONoC has great capability to perform intensive and high throughput inter-core communications required by the data exchange among cores for accelerating the parallel computing of FCNN training. The latest research progress \cite{zhang2019artificial} also indicates that more and more scientists are applying light based technologies for the research area of neural networks. 

By leveraging the advantages of ONoC, we aim to develop an acceleration model for FCNN training by addressing the following challenges:  
{\it (1) Modeling the computations and communications for training FCNN on ONoC and (2) assigning cores of ONoC to different execution stages of FCNN with the objective of minimizing the total training time.}
The specific problems include: What are the optimal numbers of cores for training a FCNN in different execution stages? How to map the neurons to cores on an ONoC for both forward propagation (FP) and back propagation (BP) within the wavelength limitation? What are the memory requirements of cores to store the FCNN parameters? 
In this paper, we address the above challenges with key contributions summarized as follows:

\begin{itemize}[leftmargin=0.1in]
\item We propose a fine-grained parallel computing model for FCNN training on ONoC, which can be used to analyze the trade-off between computation and communication in FP and BP processes. Based on this model, we derive the optimal number of cores required in each execution stage to minimize the total training time. 
\item We propose three mapping strategies for allocating the optimal numbers of cores to different stages of FCNN training. The advantages and disadvantages for each mapping strategy are discussed and analyzed in terms of hot-spot level, memory requirement, and state transitions. 
\item We evaluate our proposed acceleration schemes with extensive simulations. Firstly, we conduct simulations to show that the average prediction error on the optimal number of cores using NN benchmarks is within 2.3\%, which verifies the effectiveness of our model. Secondly, we compare our proposed methods with traditional parallel computing methods that either allocate a fixed number of cores or allocate as many cores as possible. Results show that FCNN training time by our proposed scheme on ONoC is reduced by 22.28\% and 4.91\% on average, respectively. Lastly, we evaluate the performance and energy consumption of our methods between ONoC and ENoC, which shows the training time on ONoC is reduced by 21.02\% and 12.95\% and the energy consumption is reduced by 47.85\% and 39.27\% compared with ENoC under batch sizes 64 and 128, respectively.      
	
\end{itemize}
The rest of this paper is organized as follows. Section~\ref{sec:background} presents the background with motivation examples.
Section~\ref{sec:acceleration} describes our proposed model and optimal solution.
Section~\ref{sec:mapping} illustrates three mapping strategies.  
Section~\ref{sec:evaluation} evaluates proposed models and methods, and Section~\ref{sec:related} presents related work. Finally, Section~\ref{sec:conclusion} concludes the paper.


\section{Background and Motivation}\label{sec:background}

\subsection{Fully Connected Neural Network} 
\noindent  
A FCNN has one input layer, one or more hidden layers and one output layer. The neurons of a FCNN are fully connected from one layer to the next layer with weights and bias. All neurons in the same hidden layer have the same activation function (e.g., Sigmoid, Tanh, ReLU and Softmax). Similarly, all neurons in the output layer have the same cost function (e.g., Mean squared error, Cross entropy, log-likelihood)~\cite{murray1992multilayer}. 
For each neuron, it first reads the input and executes linear computation with weights and biases, then executes the non-linear function as follows: 

\begin{equation}
\label{eq:f}
Y=A(W^TX+b) , 
\end{equation}
where  $X$ is the input vector, $Y$ is the predicted output vector,   
$W$ is the  weight vector, $b$ is the bias vector, and $A$ is the activation or cost function. The output of the neuron from the activation function is used as the
input of the neurons in the next layer. 
All neurons follow the same computation pattern layer by layer until the output layer. This is called the forward propagation in FCNN training.

After forward propagation, the predicted output and the expected output in the output layer are compared to get the loss of forward propagation. 
Then, partial derivatives are used to calculate the gradients to minimize the loss of forward propagation. 
The gradients are back-propagated to the previous layers to update the weights and biases through the chain rule computation~\cite{murray1992multilayer}. 
In back propagation, if the gradient vector calculated based on the loss is $ \rho_j$, the weights are updated as follows: 
\begin{equation}
\label{eq:b}
\rho=\sum_{j=1}^{T} \rho_j ,
\end{equation}
\begin{equation}
\label{eq:b}
W=W+\eta\rho,
\end{equation}
where $\eta$ is the learning rate and $T$ is the number of training samples. 

\vspace{-2mm}
\begin{figure}[!hpbt]
	\centering
	\includegraphics[width=8cm]{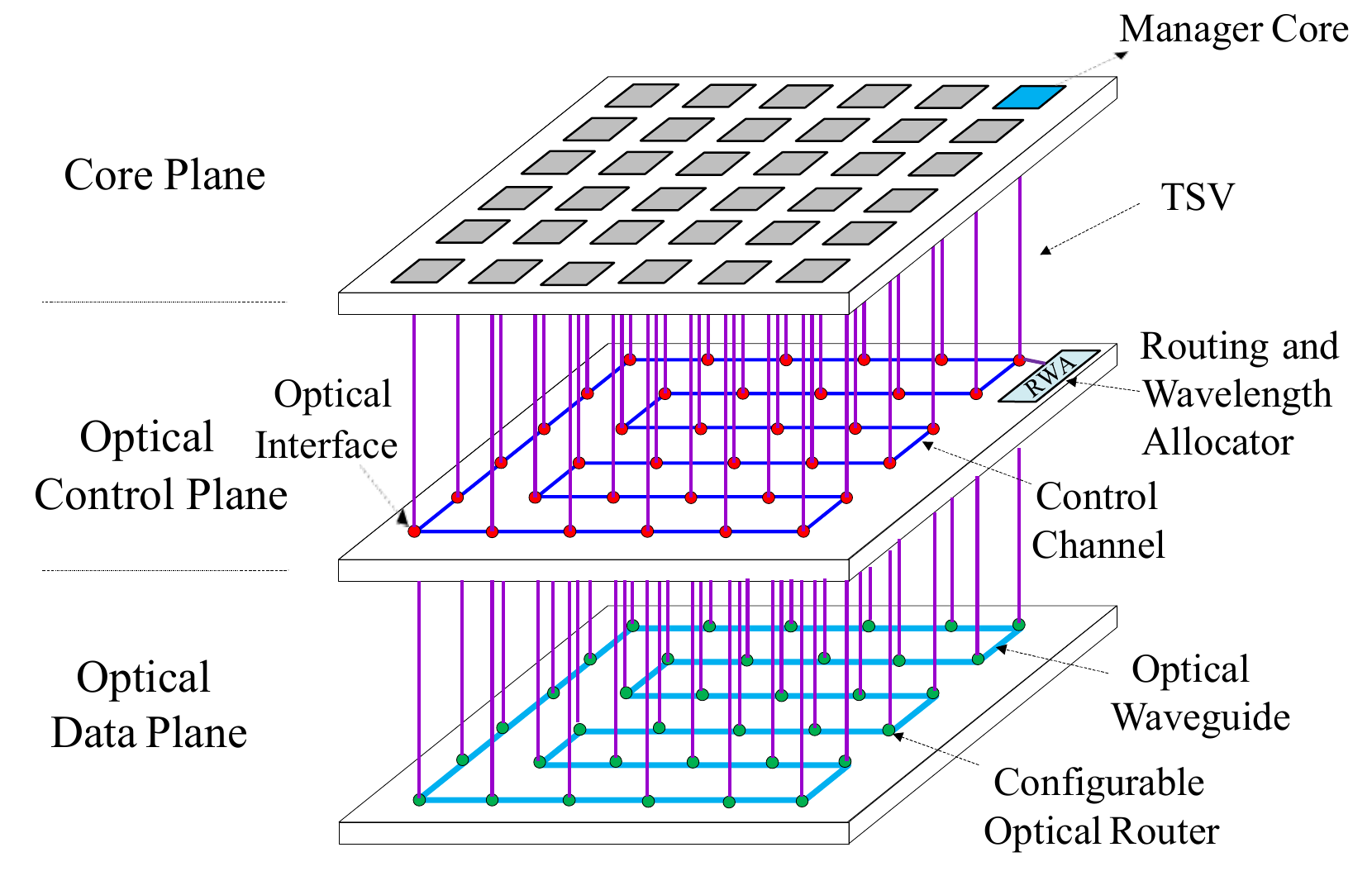}
	{
		\vspace{-2mm}
		\caption{Ring-based ONoC with three planes.}
		
		\label{fig:architecture}
	}
	\vspace{-2mm}
\end{figure}

\subsection{Optical Network-on-Chip} 
\noindent The ONoC architecture used to accelerate FCNN is shown in Fig. \ref{fig:architecture}, which contains three planes~\cite{liu2016ring}: core plane, optical control plane, and optical data plane. In general, the core plane contains the cores to realize parallel computing, while the manager core is used to calculate routing and wavelength assignment and send requests to routing and wavelength allocator (RWA) in the control plane. The optical control plane contains the RWA and a cyclic optical control channel to configure the state of cores in the core plane and the optical routers in the optical data plane, as can be seen in Fig. \ref{fig:core}. The optical data plane utilizes the configurable optical routers connected by a ring topology to provide optical data transmission, which enables cores to send and receive data packets simultaneously using WDM technology. The cores and optical routers are connected to the control channel by Through-Silicon Via (TSV) for router configuration and data transmission. Note that there are two control components: the manager core in the core plane calculates the configurations and the RWA in the optical control plane controls the network configuration. The manager core and the RWA share the same optical interface and they cooperate to transmit control packets to different interfaces by different wavelengths at the same time and each interface can receive the right packets according to the wavelength. Accordingly, the corresponding modulators in transmitters' routers and drop filters in receivers' routers are configured and ready for communications.
Fig.~\ref{fig:router} illustrates the connections for optical components in the optical router, where the splitter is used to split optical signals on the receiver side (Rx), then the activated drop filters absorb the corresponding optical signals, and finally, optical signals are converted to electrical signals by photo-detector. The coupler on the transmitter side (Tx) is used to inject modulated light into the waveguide. We assume single waveguide is used in our design with off-chip laser source, while extended work can be further investigated for multiple waveguides~\cite{ortin2017contrasting} or on-chip laser source~\cite{Batten}.
Besides, we assume that each core in the core plane has an on-chip distributed memory architecture with its L1 private cache and distributed SRAM connected to the main memory via the memory controller. The details of routing and wavelength assignment scheme and system parameters will be described in Section 4.6 and Section \ref{sec:evaluation}.

\vspace{-2mm}
\begin{figure}[!hpbt]
	\centering
	\includegraphics[scale=0.6]{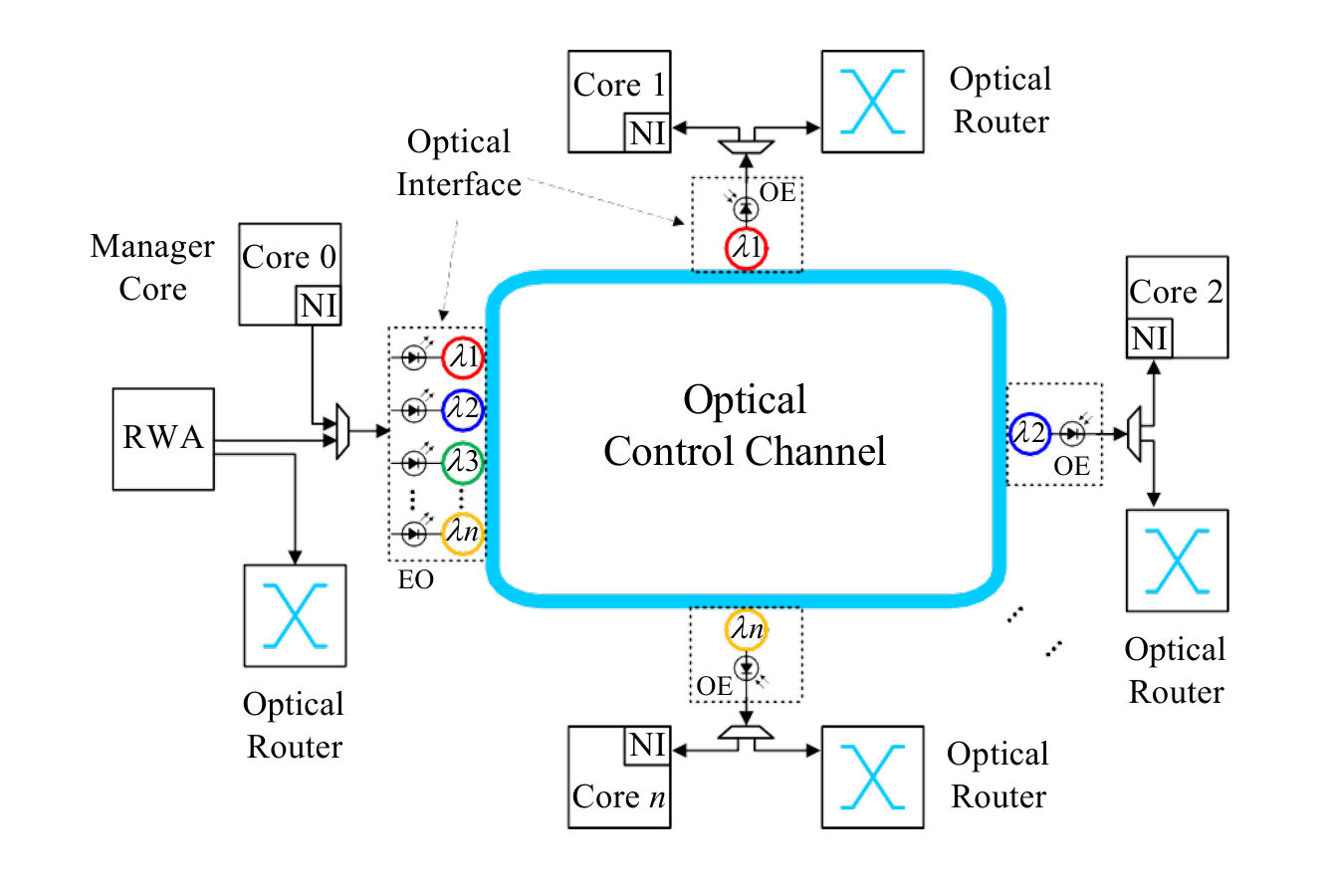}
	{
		\vspace{-3mm}
		\caption{The connection of optical interfaces in the optical control channel. 
		}
		\vspace{-4mm}
		\label{fig:core}
	}
	\vspace{1.5mm}
	
\end{figure}

\vspace{-1mm}
\begin{figure}[!hpbt]
	\centering
	\includegraphics[scale=0.3]{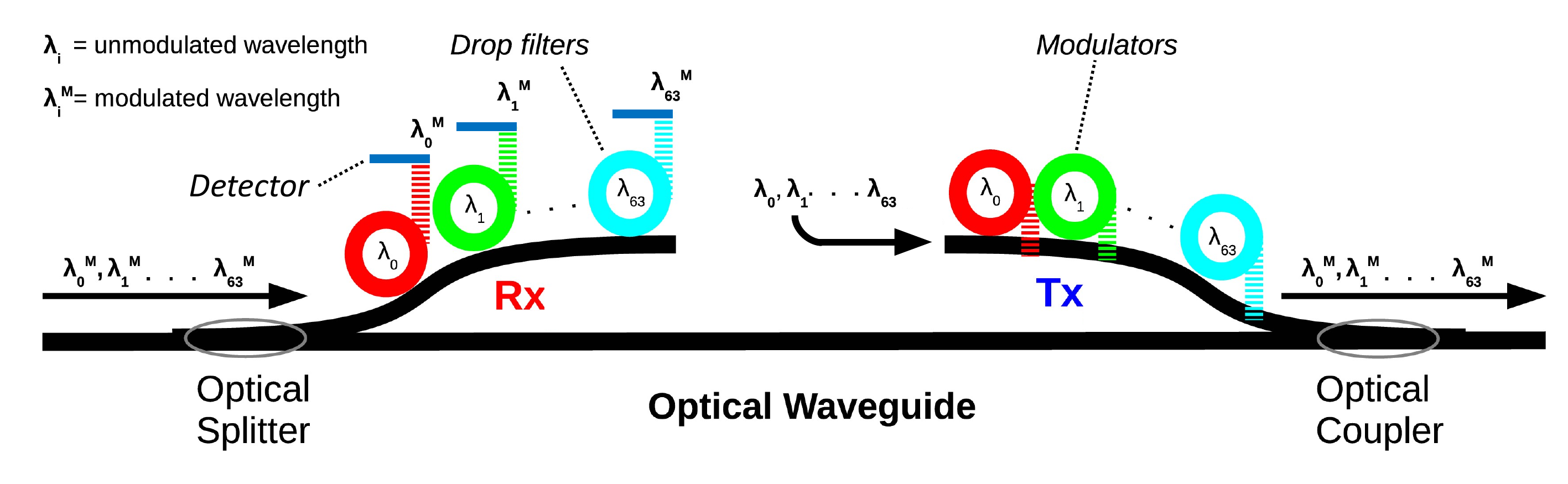}
	{
		\vspace{-2mm}
		\caption{Configurable optical router.}
		\vspace{-2mm}
		\label{fig:router}
	}
	\vspace{1.5mm}
	
\end{figure}



\vspace{-2mm}
\subsection{Motivating Examples}\label{sec:motivation}
\noindent We first explain how a FCNN can be trained in an ONoC, and then present the motivations of our solution that leverages ONoC to accelerate FCNN training. The example given in Fig. \ref {fig:neural-new} (a) is used to explain the process of FCNN training in an ONoC. 
To speed up training using parallel computation, the neurons in a FCNN can be mapped to multiple cores to execute in parallel on the ONoC, where multiple neurons can be mapped onto the same core. As illustrated in Fig. \ref {fig:neural-new} (a), one epoch of training is divided into multiple periods based on layers, and these periods are executed sequentially. 
In \textit{Period 0}, data and FCNN instructions in main memory are loaded to the distributed SRAM of ONoC cores.
In the subsequent periods, the cores mapped with neurons in the corresponding layer perform computations concurrently and then pass the outputs to the cores mapped with neurons in the next layer through inter-core communications instead of accessing main memory. 
\begin{figure}[!hpbt]
	\vspace{-1mm}
	\centering
	\includegraphics[scale=0.35]{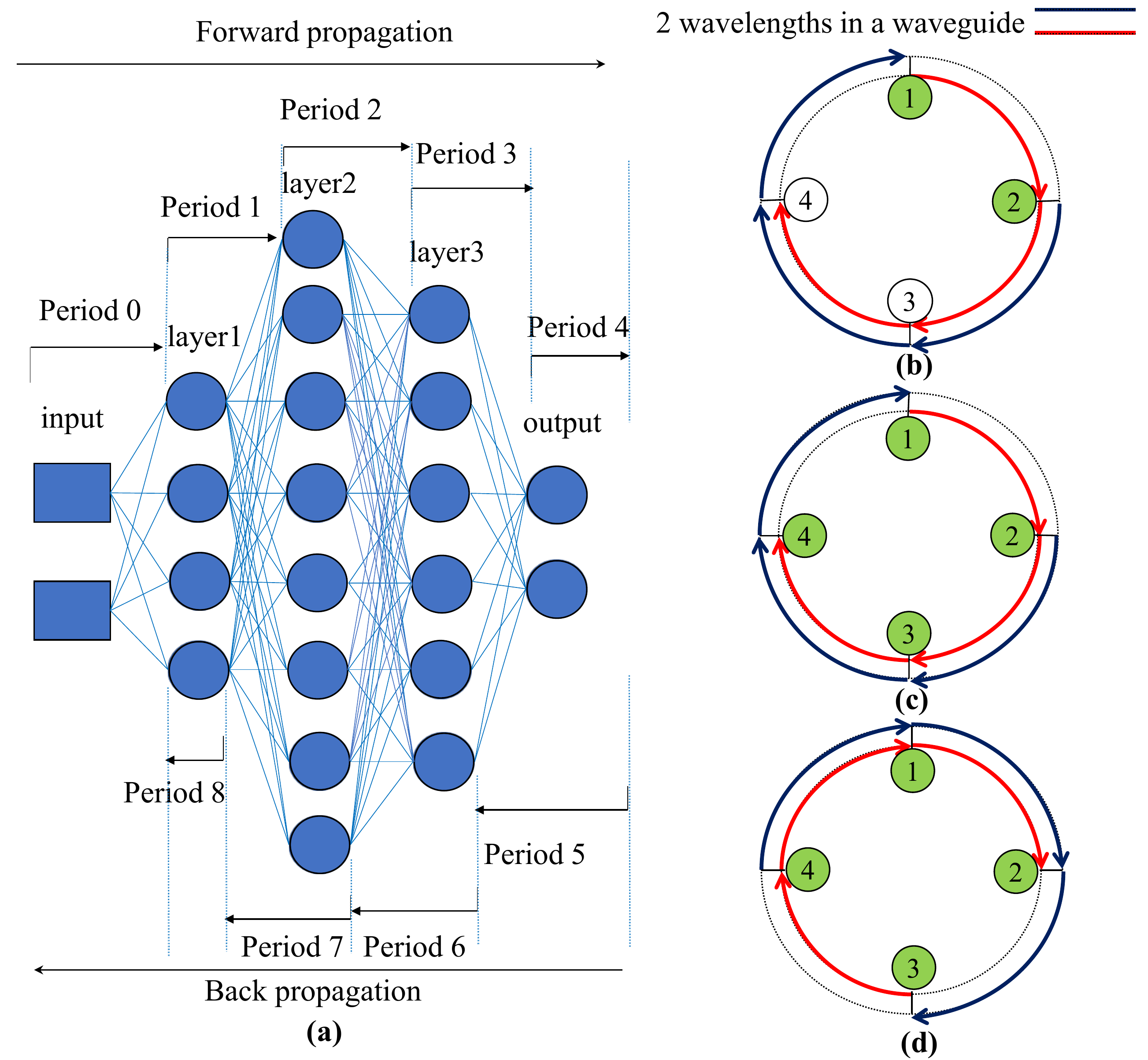}
	{
		\vspace{-1mm}
		\caption{Illustration of FCNN training on ONoC; (a) Periods of FCNN; (b) {\bf Scheme 1}: assigning 2 cores (1 and 2) for Period 1 and 4 cores (1, 2, 3, 4) for Period 2, with data in Core 1 and 2 transmitted to other cores by 2 wavelengths concurrently; (c) {\bf Scheme 2}: assigning 4 cores (1, 2, 3, 4) for both Period 1 and 2, with data in Core 1 and 2 transmitted to other cores in the first time slot; and (d) data in Core 3 and 4 transmitted to other cores in the second time slot.}
		\label{fig:neural-new}
	}
\vspace{-3mm}
\end{figure}

With such a mapping of FCNN training on ONoC, we illustrate our motivations by the following examples:

\noindent\textbf{Example I}: This example is to illustrate \textit{FCNN training can be accelerated by exploiting optical transmission for data exchange between cores in adjacent periods}. When a FCNN is trained in a traditional multi-core system, the outputs of each neuron, weights, biases and gradients in the current execution stage may be written to the main memory for further processing in the subsequent execution stage. When the neurons in two adjacent periods are not mapped to the same core, the data needed for the computation in the next period may be loaded from the main memory to the corresponding cores. If the cores are reading the same parameters in the main memory, more delay will occur due to bank or port conflict in the main memory. 
These frequent read and write operations in main memory will cause relatively longer delay, thereby slowing down the FCNN training process with high energy cost. 
For instance, it takes around 180 cycles to access main memory in three level cache intel i7 CPU~\cite{levinthal2009performance} and takes about 30 cycles and 375 cycles in the shared memory and main memory of NVIDIA Volta GPU architecture~\cite{jia2018dissecting}. 
However, the off-chip memory access for neural network training can be alleviated by 94.1\% on average using on-chip inter-core communication~\cite{chen2019noc}.
The average latency of inter-core communication ranges from 19 to 91 cycles in ENoC~\cite{catania2006methodology} and 3 to 7 cycles in ONoC~\cite{morris20123d}, which are much faster than main memory access. Obviously, in comparison with ENoC, ONoC has much more potential to reduce the amount of time required for data communication during FCNN training.
Moreover, by leveraging the WDM technology, multiple cores can transmit and receive data concurrently through waveguide as long as they use different wavelengths. 
This can further reduce the amount of time for data communications between neurons in adjacent layers. 

\vspace{0 mm}
\noindent\textbf{Example II}: This example is to illustrate \textit{there is a trade-off between computation and communication for FCNN training on ONoC, which allows the use of optimization techniques to further reduce the amount of training time}. Suppose that the FCNN given in Fig. \ref {fig:neural-new} (a) is trained on an ONoC with four cores connected by ring topology, and assume that the number of available wavelengths is 2. 
We use the execution of \textit{Period 1} and \textit{Period 2} as an example, and consider the following two mapping schemes between neurons and cores: (1) \textbf{Scheme 1}: \textit{assigning 2 cores for Period 1 and 4 cores for Period 2}. The four neurons in \textit{Period 1} are mapped to \textit{Core 1} and \textit{2} with each core having two neurons, and the eight neurons in \textit{Period 2} are mapped to \textit{Core 1, 2, 3} and \textit{4} with each core having two neurons. Since only 2 wavelengths are available, \textit{Core 1} and \textit{Core 2} can send their data to the other cores in parallel using 2 wavelengths (represented by the red and blue curved arrows), as illustrated in Fig. \ref {fig:neural-new} (b). (2) \textbf{Scheme 2}:  \textit{assigning 4 cores for both Period 1 and Period 2}.  Since only two cores can send in parallel at one time, TDM has to be used costing two time slots. In the first time slot, as illustrated in Fig. \ref {fig:neural-new} (c), \textit{Core 1} and \textit{Core 2} can send their data to other cores in parallel using 2 wavelengths; In the second time slot, \textit{Core 3} and \textit{Core 4} can send their data to the other cores in parallel reusing the 2 wavelengths, as illustrated in Fig. \ref {fig:neural-new} (d). Hence, Scheme 1 uses fewer cores in \textit{Period 1} spending more time on computation with less time on communication than Scheme 2.  During the back-propagation, cores can be reused according to a data locality constraint (to be defined in Eqs. (11) in Section 3.1.3), so that each neuron is mapped
to the same core in the FP and BP to maximize data locality. For example, the cores assigned to \textit{Period}
1 in FP process (e.g. \textit{Core 1} and \textit{2} in Scheme 1) are the same as that assigned to Period 8 in BP process, while the cores assigned to \textit{Period 2} in FP process (e.g. \textit{Core 1, 2, 3} and \textit{4} in Scheme 1) are the same as that assigned to \textit{Period 7} in BP process. Therefore, the senders (\textit{Core 1} and \textit{2}) in \textit{Period 1} become the receivers in Period 8 and the receivers (\textit{Core 1, 2, 3 and 4}) in \textit{Period 2} become the senders in \textit{Period 7}. 
It can be seen that using more cores can reduce the computation time, but incur more time on communication, which shows
there is a trade-off between computation and communication in the mapping between neurons and cores. 
 \textit{The challenging problem is: to train a given FCNN on an ONoC with a limited number of available wavelengths,  how to map the neurons to cores for different periods so that the total amount of time for FCNN training is minimized.} 
To tackle this challenge, we propose a FCNN acceleration model for ONoC that explores the trade-off between computation and communication to calculate the optimal number of cores for each period and minimize the FCNN training time.

\vspace{-1 mm}
\section{FCNN Acceleration Schemes on ONoC}\label{sec:acceleration}
In this section, we first present our FCNN accelerating model, and then derive the optimal number of cores for each period to minimize the total training time. 


\vspace{-2mm}
\subsection{Parallel Accelerating Model on ONoC} \label{sec:model}

\noindent We consider the training of a FCNN with $l+1$ layers labeled from layer 0 to layer $l$, where layer 0 and layer $l$ are the input and output layer, respectively. The number of neurons in layer $i$ is represented by $n_i$ for $i\in[0, l]$. As illustrated in Fig.~\ref {fig:neural-new} (a), one epoch of training is divided into multiple periods based on layers. The FP process is divided into $l+1$ periods labeled from \textit{Period 0} to \textit{Period $l$}, and the BP process is divided into another $l$ periods labeled from \textit{Period $l+1$} to \textit{Period $2l$}. We assume the ONoC contains $m$ cores connected by ring topology, and the number of available wavelengths is $\lambda_{max}$. We focus on the acceleration of one epoch during FCNN training because each epoch is repetitive. Note that all FCNN parameters and intermediate values are stored in SRAM of the corresponding cores distributively, with these parameters staying in the corresponding SRAM during one epoch of training. Cores used in different layers  exchange data by optical communications in ONoC.  


\subsubsection {Computation Cost} 
We use $m_i$ to represent the number of cores assigned to \textit{Period i} for parallel computation, and assume that the neurons are evenly mapped to the $m_i$ cores in each period. In the FP process, \textit{Period 0} is used to map neurons to cores and load the training data to the corresponding cores from the main memory. Hence, computation starts from \textit{Period 1}.
According to the definition of periods, the neurons in layer $i$ where $i\in[1, l]$ get involved in \textit{Period i} in the FP process. The neurons in layer $2l-i+1$ where $i\in[l+1, 2l]$ get involved in \textit{Period i} during the BP process. Therefore, the corresponding neuron number $n_{i}$ in FP process is the same as $n_{2l-i+1}$ in the BP process. Let $X_i$ be the number of neurons mapped to each core in \textit{Period i} during the FCNN training. We have 
\begin{equation}
X_{i}=\begin{cases} ~~\left\lceil\frac{n_{i}}{m_{i}}\right\rceil , i\in[1,l] ;\\ \left\lceil\frac{n_{2l-i+1}}{m_{i}}\right\rceil , i\in[l+1,2l] .\end{cases} \label{eq:neuron}
\end{equation}

Assume that all cores are homogeneous with the same computation capacity $C$. 
In the FP process, the computations consist of multiply-accumulate operations in the corresponding functions at each layer. We use $\alpha_i$ to represent the amount of each neuron computation in \textit{Period i} of the FP process. When the batch size (i.e., the number of samples in one training epoch) is larger than one, $\alpha_i$ is the amount of computation for each neuron in \textit{Period i}  to process all samples in the current training. Then the amount of computation time at each core in \textit{Period i} of the FP process is $\frac{\alpha_i X_i}{C}$. In the BP process, the computation is dominated by the updates on weight and bias that consists of differentiation operations. 
For each neuron in \textit{Period i} of the BP process, it needs to update the weights for the connections to all neurons in \textit{Period i-1}. We use $\beta_i$ to represent the amount of computation to update the weight of one connection based on all training samples according to Eqs. (2) and (3). 
Since the gradient calculation related to the neurons in the next layer, the amount of computation updating weight at each core in \textit{Period i} of the BP process is $\beta_iX_{i}n_{2l-i}$. Unlike weight update, the bias is updated on a per-neuron basis. As updating a weight and updating a bias have the same complexity, the amount of computation for bias update in \textit{Period i} of the BP process is $\beta_iX_{i}$.  Therefore, the computation time cost for \textit{Period i} during the BP process is $\frac{\beta_i X_{i} (n_{2l-i}+1)}{C}$. 

Let $f(m_{i})$ represent the amount of  computation time required for each of the $m_i$ cores in \textit{Period i}. We have 
\begin{equation}
f(m_i)=\begin{cases} ~~\frac{\alpha_i X_i}{C}, i\in[1,l] ;\\ \frac{\beta_i X_{i} (n_{2l-i}+1)}{C} , i\in[l+1,2l] .\end{cases} \label{eq:f1_function_b}
\end{equation}

\vspace{-1mm}
\subsubsection {Communication Cost}
In the FP process, the outputs of the neurons in layer $i$ need to be propagated to neurons in layer $i+1$, whereas  the gradients computed at each neuron in layer $i$ need to be back propagated to all neurons in layer $i-1$ in the BP process. As illustrated in Fig. \ref {fig:neural-new} (b), the propagation of neuron outputs and gradients can be effectively implemented in a ring-based ONoC using the broadcast operation, that is, if one core in the current period sends data along the ring, the cores assigned to the next period can receive the data by filtering a small portion of the transmitted optical signal. By leveraging the WDM technology, the communications in each period can be parallelized by letting multiple cores transmit simultaneously using different wavelengths.  For \textit{Period i}  that demands communications, all the $m_i$ cores can transmit concurrently if $m_i\leq \lambda_{max}$; otherwise TDM needs to be used to complete the transmissions from the $m_i$ cores. We use $B_i$ to denote the amount of time for one core in \textit{Period i} to complete the communications, which can be calculated based on the amount of data to be transmitted, the read and write operations for cache access, and the O/E and E/O conversions. Let $g(m_{i})$ be the total amount of time required to complete communications in \textit{Period i}. We have
\begin{equation}
g(m_i)=\begin{cases} 0,\quad\qquad\qquad i=1, l \textrm{ and }2l ;\\ \left\lceil\frac{m_{i}}{\lambda_{max}}\right\rceil B_i, \quad\textrm{otherwise} .\end{cases} \label{eq:communicationcost}
\end{equation}
 $g(m_{1})=g(m_{l})=g(m_{2l})=0$ because there is no communication in these periods. 

\vspace{1mm}
\subsubsection {Problem Formulation} \label{sec:formulation}
Our objective is to compute the optimal number of cores allocated to each period, so that the training time required for one epoch can be minimized. The total amount of time to complete one epoch training, denoted by $T$, includes the time for FP ($Time_{fp}$) and the time for BP ($Time_{bp}$). Based on our acceleration model, it can be calculated as follows:
\begin{equation}\label{eq:objectivefunction}
\begin{aligned}
T&=Time_{fp} + Time_{bp} \\
&=D_{input}+\sum_{i=1}^{2l}\Big(f(m_i)+g(m_i)+\zeta_{i}\Big) ,
\end{aligned}
\end{equation}
where $D_{input}$ represents the time delay caused by loading input data and FCNN instructions from the main memory to the assigned cores in \textit{Period 0}, and $\zeta_i$  represents the additional delay caused by extra main memory access, software overhead, synchronization, and etc. in the $i$th period. We formulate our problem as an optimization problem as follows: 
\begin{equation}
\label{eq:objective}
\textbf{Minimize } ~T,
\end{equation}
and the optimization subjects to the following constraints. 

\begin{itemize}
\item {\bf Constraint on the size of ONoC:} The number of assigned cores $m_i$ in each period is no larger than $\phi m$:
\begin{equation}
\label{eq:c1}
   m_i \leq \phi m ,  i\in[1,2l], \phi\in(0, 1] , \\
\end{equation}
where $\phi$ is used to control the utilization of cores according to the system limitations such as signal crosstalk and power loss. As higher utilization of cores requires longer optical transmission path through more optical elements leading to higher crosstalk and power loss, $\phi$ can be predefined according to the ONoC system as a threshold to satisfy the crosstalk and power loss limitations.
\item {\bf Constraint on the size of FCNN:} The number of assigned cores in each period cannot be larger than the number of the neurons in that period:
\begin{equation}
\label{eq:c2}
m_i \leq n_{i} , i\in[1,l] \textrm{ and }  m_i \leq n_{2l-i+1} , i\in[l+1,2l] .
\end{equation}
\item {\bf Constraint on data locality:} each neuron is mapped to the same core in the FP and BP to maximize data locality. Hence, the number of cores assigned to \textit{Period i} in FP process is the same as that assigned to \textit{Period 2l-i+1} in BP process:
\begin{equation}
\label{eq:c4}
    m_{2l-i+1}=m_i  , i\in[1,l] . 
\end{equation} 
\end{itemize}

\vspace{-1mm}

\subsection{Optimal Solution} \label{sec:random}

\noindent In this section, we derive the optimal solution for the optimization problem formulated in Eq. (\ref{eq:objective}).
%
%
\begin{lemma}\label{lemma:optm}
	The optimal number of cores allocated to \textit{Period i}, denoted by $m^{*}_i$  is 
\begin{align}\nonumber
m^{*}_i=\begin{cases} \min\{\lceil\sqrt{\frac{\theta_i}{B_iC}}\rceil, \phi m\}, & i=1;\\  \min\{\lceil\sqrt{\frac{\theta_i}{(B_i+B_{2l-i+1})C}}\rceil, \phi m\}, &  1< i < l;\\ \min\{\lceil\sqrt{\frac{\theta_i}{B_{i+1}C}}\rceil, \phi m\}, & i=l;\end{cases} \label{eqn:ip_scheduling_variables}
\end{align}
where $\theta_i = n_{i}\lambda_{max}[\beta_{2l-i+1}(n_{i-1}+1)+\alpha_{i}]$.
\end{lemma}
\begin{proof}
According to Eq. (\ref{eq:objectivefunction}), $T$ is a multi-variable function with variables $m_1$, $m_2$, $\cdots$, $m_l$. Assuming each $m_i$ where $i\in[1,l]$ is real number and the ceiling operators in all equations are removed, $T$ becomes a continuous function. Hence, $T$ is minimized when 
\begin{equation}\label{lemma:proof}
	\frac{\partial T}{\partial m_i} =0, \forall i\in[1, l].
\end{equation}
\textbf{Case I}: When $i=1$, only $f(m_1)$, $f(m_{2l})$, and $g(m_1)$ are functions of $m_1$ and $f(m_{2l})=0$. According to Eqs. (\ref{eq:f1_function_b}), (\ref{eq:communicationcost}) and (\ref{eq:objectivefunction}), 
\begin{equation}
\begin{aligned}\nonumber
&\frac{\partial T}{\partial m_1} =\frac{\partial f(m_1)}{\partial m_1}+\frac{\partial f(m_{2l})}{\partial m_1}+\frac{\partial g(m_1)}{\partial m_1}\\
&= -\frac{n_1[\alpha_1+\beta_{2l}(n_0+1)]}{Cm^2_1} +\frac{B_1}{\lambda_{\max}} .
\end{aligned}
\end{equation}
Let $\frac{\partial T}{\partial m_i} =0$, then $T$ is minimized when 
\begin{equation} m_1=\sqrt{\frac{n_{1}\lambda_{max}\big[\beta_{2l}(n_0+1)+\alpha_1\big]}{B_1C}}.
\end{equation}
\textbf{Case II}: When $2\leq i\leq l-1$, only $f(m_i)$, $f(m_{2l-i+1})$, $g(m_i)$ and $g(m_{2l-i+1})$ are functions of $m_i$. Similarly,
let $\frac{\partial T}{\partial m_i} =0$, then $T$ is minimized when
\begin{equation} m_i=\sqrt{\frac{n_{i}\lambda_{max}\big[\beta_{2l-i+1}(n_{2l-i}+1)+\alpha_i\big]}{(B_i+B_{2l-i+1}) C}} .
\end{equation}
\textbf{Case III}: When $i=l$, only $f(m_l)$, $f(m_{l+1})$, and $g(m_{l+1})$ are functions of $m_l$ and $g(m_{l})=0$. Similarly, $T$ is minimized when
\begin{equation} m_l=\sqrt{\frac{n_{l}\lambda_{max}\big[\beta_{l+1}(n_{l-1}+1)+\alpha_l\big]}{B_{l+1} C}} .
\end{equation}
From the above three cases, we use $\theta_i$ to represent the numerator of $m_i$ so that we have $\theta_i = n_{i}\lambda_{max}[\beta_{2l-i+1}(n_{i-1}+1)+\alpha_{i}]$.
Since $m_i$ for $i\in[1,l]$ must be an integer and it must be no larger than $\phi m$ according to the constraint given in (\ref{eq:c1}),  $m^*_{i} =\min(\lceil m_i \rceil, \phi m)$. Hence, this lemma holds. 
\end{proof}

\begin{theorem}
	Given a FCNN with $l+1$ layers with layer $i$ having $n_i$ neurons, the minimum amount of time to perform one epoch training on an ONoC with $m$ cores and $\lambda_{\max}$ wavelengths is $T^* =D_{input}+\sum_{i=1}^{2l}\Big(f(m^*_i)+g(m^*_i)+\zeta_{i}\Big)$.
\end{theorem}
\begin{proof}
According to the proof for Lemma \ref{lemma:optm}, the amount of time for one epoch training $T$ is minimized when $m_i=m^*_i$ for each $i\in[1,2l]$. Hence, by replacing $m_i$ with $m^*_i$ in Eq. (\ref{eq:objectivefunction}), the minimum amount of time for epoch training is 
$T^* =D_{input}+\sum_{i=1}^{2l}\Big(f(m^*_i)+g(m^*_i)+\zeta_{i}\Big)$.
%
\end{proof}

\vspace{-3mm}

\section{Allocation of Cores on ONoC}\label{sec:mapping}
\noindent After the optimal number of cores required for each period is derived, the next step is to investigate the core allocation for different periods on ONoC. 
In the following, we first present three mapping strategies and then discuss their advantages and disadvantages with the analysis of hotspot level, memory requirement and state transitions.  



\begin{figure}[!hpbt]
	\vspace{-2mm}
	\centering
	\includegraphics[scale=0.27]{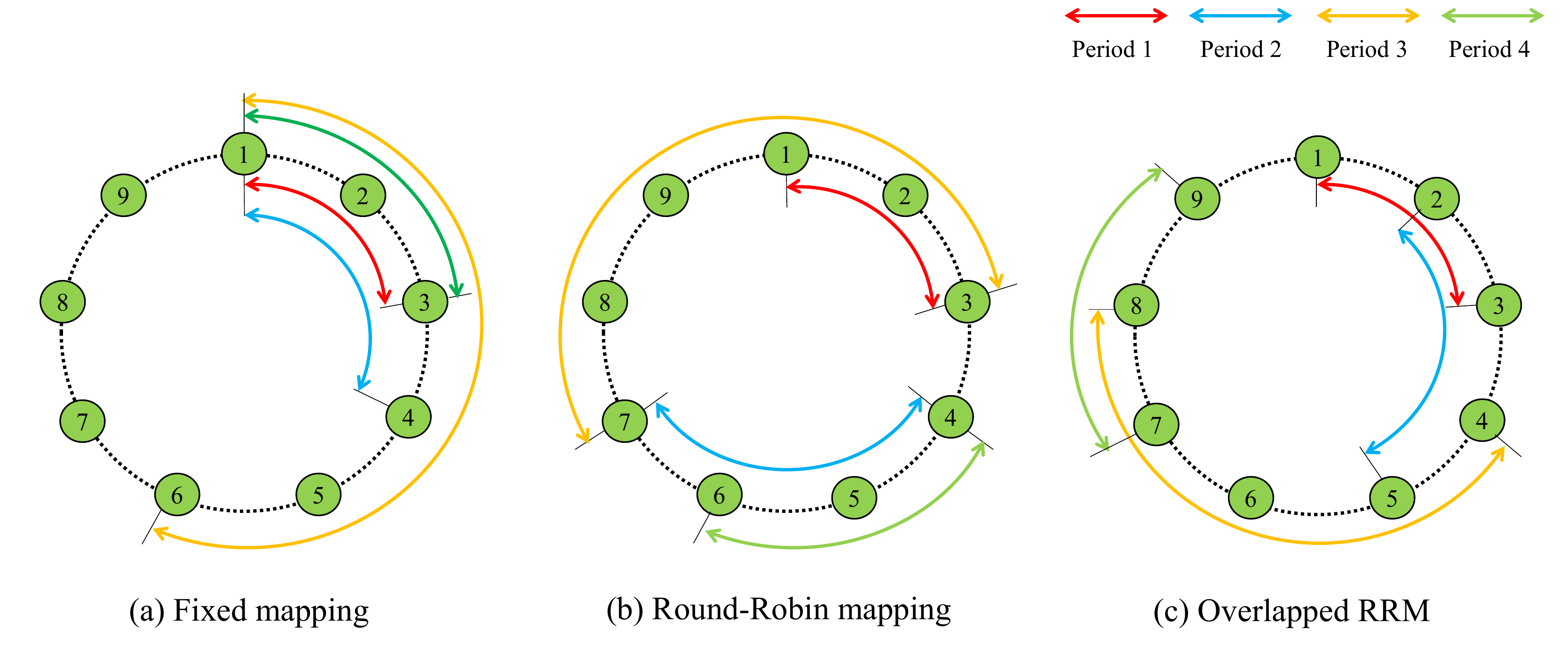}
	{
		\vspace{-4mm}
		\caption{Three mapping strategies}
		\label{fig:three_mapping}
	}
\vspace{-3mm}
\end{figure}
\vspace{-2mm}

\subsection{Three Mapping Strategies}

\noindent \textbf{Strategy I: Fixed Mapping (FM).} The cores for each \textit{Period i} are mapped sequentially along the ring in clockwise order, always starting from a fixed core (e.g. $core_1$). That is, the cores allocated to \textit{period i} is  $[core_1, core_2, ... , core_{m^*_i}]$. For example, for a 5-layer FCNN to be trained on an ONoC with 9 cores, assume $m^*_1=3$, $m^*_2=4$, $m^*_3=5$, and $m^*_4=3$ according to our model. 
As illustrated in Fig. \ref{fig:three_mapping} (a), $[core_1, core_2, core_3]$ are assigned to \textit{Period 1},  $[core_1, core_2, core_3, core_4]$ are assigned to \textit{Period 2},   $[core_1, core_2, core_3, core_4, core_5]$ are assigned to \textit{Period 3}, and $[core_1, core_2, core_3]$ are assigned to \textit{Period 4}. 
The advantages of this strategy include: (1) \textit{fewer messages to receive} due to the core reuse, even though the number of transmitted messages remains unchanged. For example, $core_1$, $core_2$, and $core_3$ are reused in \textit{Period 1} and \textit{Period 2}. Each time one of these three cores transmit a message, only three of the cores assigned to \textit{Period 2} need to receive the message. If there is no core reuse, all four cores assigned to \textit{Period 2} need to receive the message; (2) \textit{less energy consumption} due to the low frequency to change state of the cores and optical routers. For example, $core_1$, $core_2$ and $core_3$ are used for all periods, and they just need to change state one time in each epoch training;
(3) \textit{less crosstalk and insertion loss} since the maximum path length during the whole FCNN training is $\max_{i=1}^{l}m^*_i-1$ (e.g. 5 in the example).  However, this method also has several disadvantages: (1) \textit{high memory requirement} because the reused cores need to store the parameters for all involved periods; (2) \textit{hot-spot and unbalanced thermal dissipation} because some cores keep working during the training process.   

\noindent \textbf{Strategy II: Round-Robin Mapping (RRM).} The cores for each \textit{Period i} are mapped sequentially along the ring, starting from the core next to the last core allocated to \textit{Period i-1}. As illustrated in Fig. \ref{fig:three_mapping} (b), $[core_1, core_2, core_{3}]$ are assigned to first period, $[core_{4}, core_{5}, core_{6}, core_{7}]$ are assigned to the second period, and so on. Different from FM, there are no reused cores in adjacent periods in RRM. The merits of the RRM are: (1) \textit{hot-spot avoidance and balanced thermal dissipation} due to the nature of round-robin selection; (2)\textit{ less memory requirement} since the cores are involved in fewer periods, thus only the parameters of the involved periods are stored. However, it also has some drawbacks including: (1) \textit{more messages to receive} since there is no core reuse between adjacent period; (2)\textit{ more energy consumption} due to the frequent state changes at the cores and optical routers. For example, different from FM, $[core_1, core_2, core_{3}]$ in RRM are active before \textit{Period 1} starts, and then become idle after \textit{Period 1} is finished. $[core_{4}, core_{5}, core_{6}, core_{7}]$ are active before \textit{Period 2} and receives the message from $[core_1, core_2, core_{3}]$, then become idle in the end of \textit{Period 2} and so on.  

\noindent \textbf{Strategy III: Overlapped Round-Robin Mapping (ORRM).} To overcome the drawbacks of both FM and RRM, we propose a third mapping scheme called Overlapped Round-Robin Mapping, which is similar to RRM strategy but allows core reuse in the adjacent periods. As illustrated in Fig. {\ref{fig:three_mapping}} (c), $[core_1, core_2, core_{3}]$ are assigned to \textit{Period 1} and $[core_{2}, core_{3}, core_{4}, core_{5}]$ are assigned to \textit{Period 2}, where $core_{2}$ and $core_{3}$ are reused in these two periods.  

Let $r_i$ be the number of cores reused in \textit{Period i} and \textit{Peirod i-1}. To balance core reuse, we define the expected number of reused cores in two adjacent periods, denoted by $E[r]$, as
 \begin{equation}
E[r] = \begin{cases}
 0, \qquad \qquad \textrm{ if }  \sum_1^{l}{m^*_i} \leq m;\\
\frac{\sum_1^{l}{m^*_i}-m}{l-1},  \textrm{ otherwise}. \end{cases} \label{eq:ol}
\end{equation}
When $\sum_1^{l}{m^*_i} \leq m$, no core will be reused to reduce the memory requirement, which is equivalent to RRM strategy. When $\sum_1^{l}{m^*_i} > m$, some cores are reused, which are expected to be evenly distributed in different periods to balance hot-spot and thermal dissipation. Based on only $E[r]$, it might occur that the expected number of cores to be reused in \textit{Period i} is larger than the optimal number of cores allocated to \textit{Period i} and a core reused in \textit{Period i}  may be further reused in the following periods.
To deal with these issues, we define reused core number $r_i$ as follows: 
 \begin{equation}
\label{eq:shift}
r_i= \min\big(round(E[r]), (m^*_{i-1}-r_{i-1}), m^*_{i}\big), i \in [2,l];
\end{equation}
where $r_1 = 0$. Let $id_i$ be the index of the first core assigned to \textit{Period i} and set $id_1=1$. Since cores are allocated to \textit{Period i} sequentially along the ring, we have 
\begin{equation}
\label{eq:id}
id_i = 1+\sum^{i}_{2}(m^*_{i-1}-r_i), i \in [2,l].
\end{equation}

\begin{lemma}\label{lemma2}
	In ORRM, the number of consecutive periods that a core can run is at most four in one epoch FCNN training, when $m^*_i+m^*_{i+1}-r_{i+1}\leq m, \forall i \in [1,l-1]$. 
\end{lemma}
\begin{proof}
	In the FP training process, when $m^*_i+m^*_{i+1}-r_{i+1}\leq m, \forall i \in [1,l-1]$, any two adjacent periods are allocated to the cores without exceeding one circular round of the ring. It can be easily proven that there is no core running for more than 2 consecutive periods in the FP process. While the cores used in the first period of BP are reused in the last period of FP, the maximum periods that a core can run consecutively is at most 4 in one epoch FCNN training.    
\end{proof}

The pseudo-code of ORRM is given in Algorithm \ref{alg:mapping}. The input includes the optimal number of cores for each period ($m^*_i$, $i\in[1,l]$) and the total number of cores in the ONoC ($m$). The output is a mapping matrix where $M(i,j,k)=1$ if the $j^{th}$ neuron in the $i^{th}$ layer is mapped to core $k$. 
Line (3) and line (8) are executing the mapping of the neurons in each period evenly to the cores assigned to that period.
 \begin{algorithm}
	\SetKwData{Left}{left}
	\SetKwData{This}{this}
	\SetKwData{Up}{up}
	\SetKwFunction{Union}{Union}
	\SetKwFunction{FindCompress}{FindCompress}
	\SetKwInOut{Input}{input}
	\SetKwInOut{Output}{output}
	\caption{Overlapped Round Robin Mapping } \label{alg:mapping}
	\Input{ $m^*_i$ where $i\in [1, l]$  and $m$}
	\Output{$M(i,j,k)$ where $i\!\in\! [1, l]$, $j\in[1, n_i]$, $k\in[1,m]$} 
	\BlankLine
	$sum  = \sum_{i=1}^{l} m^*_i$\;
	Calculate $E[r]$ according to Eq. (\ref{eq:ol}).\\
	Assign [$core_1$, $core_2$, … , $core_{m^*_1}$] to \textit{Period 1};\\
	$k=\lceil \frac{j}{X_1}\rceil$ ;\\
    $M[1, j,k] =1$, where $j\in[1,n_1]$ ;\\
	\For{$i\leftarrow 2$ \KwTo $l$}{
		Calculate $id_i$ according to Eq. (\ref{eq:id})\\
		Assign [$core_{id_i~mod~m}$, $core_{(id_i+ 1)~mod~m}$, $\cdots$, $core_{(id_i + m^*_i-1)~mod~m}$] to \textit{Period i};\\
		$k=\lceil \frac{j}{X_i}\rceil$ ;\\
		$M[i, j,k] =1$, where $j\in[1,n_i]$;
	}
\end{algorithm}

It can be seen that ORRM is a compromised method of FM and RRM. Therefore, ORRM has merits of these two strategies: (1) \textit{less energy consumption and fewer messages to receive}; (2) \textit{less memory requirement}; (3) \textit{hot-spot alleviating and thermal dissipation balancing}. We provide the analysis details as follows. 
%
%
%
%
%
%
%
%
\vspace{-2mm}
\subsection{Hotspot Analysis}
\noindent Hotspots can be caused by non-uniform workload distribution among cores where some cores need to handle relatively higher workload on computation and communication than others. Hence, we analyze hotspots by comparing the maximum number of consecutive periods for cores to run in three mapping strategies.   


\begin{theorem}\label{theorem2}
	The maximum number of consecutive periods for cores to run during one epoch training (FP and BP) is 
	\vspace{-1mm}
	\begin{itemize}
		\item $2l$ in Fixed Mapping;
		\item  $2$ in  Round-Robin Mapping when $m^*_i+m^*_{i+1}\leq m,  \forall i \in [1,l-1]$;
		\item $4$ in ORRM when $m^*_i+m^*_{i+1}-r_{i+1}\leq m, \forall i \in [1,l-1]$.
	\end{itemize}	 
\end{theorem}
\begin{proof}
	
	Fixed Mapping has relatively severe hot-spot situation because there are $\min_{i=1}^{l}m^*_i$ cores that keep running from \textit{Period 1} to \textit{Period 2l}, where there are $m-\max_{i=1}^{l}m^*_i$ cores that keep idle. This can lead to hot-spots with unbalanced thermal dissipation. 
	For Round-Robin Mapping, cores are not reused by two adjacent periods if the number of cores assigned to these two periods does not exceed the total number of cores in the ring (i.e. $ m^*_i+m^*_{i+1}\leq m,  \forall i \in [1,l-1]$), except for those cores running 2 consecutive periods in last period of FP reused in the first period of BP.  
	This round-robin fashion allows the cores to take turns to be active and idle, which maintains relatively balanced thermal dissipation and prevents the occurrence of hot-spots to some extent. 
	For ORRM, cores are reused only in two adjacent periods in FP and BP processes with cores reused at the end of FP and beginning of BP running for at most four consecutive periods according to Lemma \ref{lemma2}. 
	Therefore, the difference of active time among the cores is at most three periods no matter how many periods during the training process, which is not easy to form hot-spots with relatively balanced thermal dissipation.
\end{proof}
\subsection{State Transition Analysis}
\noindent When a FCNN is trained on an ONoC, some cores and their associated optical routers need to stay active for some periods and turn to idle in other periods. Frequent transitions between the active state and idle state can degrade the performance and energy efficiency. 
Therefore, we estimate the number of state transitions in one FCNN training epoch for the three mapping strategies ranking from low (1) to high (3) as in Table \ref{tab:transit}:
\begin{table}[!ht]
	\vspace{-1mm}
	\captionsetup{justification=centering}
	\caption{State Transition Numbers for Three Mapping Methods}
	\label{tab:transit}
	\centering
	\def\arraystretch{2}
	\setlength{\tabcolsep}{1mm}
	\begin{tabular}{ c  c  c}
		\hline 
		\! Name of mapping \! & \! Number of state transitions \! & \! Rank \\
		\hline 
		\! FM \! & \!$2(m^*_1+\sum^{l}_{i=2}|m^{*}_{i}-m^{*}_{i-1}|)$\! & \! 1\\
		
		\! ORRM \!& \!$2(\sum^{2l}_{1}m^{*}_{i}-m^{*}_l-\sum^{2l}_{2}r_{i})$\! & \! 2\\
		
		\! RRM \! & \!$2(\sum^{2l}_{1}m^{*}_i-m^{*}_l)$\! & \! 3\\
		
		\hline
	\end{tabular}
\end{table} 

\subsection{Crosstalk and Insertion Loss}
Insertion loss and signal crosstalk are critical issues for ONoC design. For a specific optical routing path, the insertion loss~$I\!L$~ can be calculated as
\begin{equation}
	\label{eq:il}
	I\!L\!=\!I\!L_{l}\!\times\!(N_{r}\!-\!1)\!+\!I\!L_{r}\!\times\!N_{r}\!+\!I\!L_{eo}\!+\!I\!L_{oe}, 
\end{equation}
where $N_{r}$ is the total number of optical routers in the path, $I\!L_{l}$ and $I\!L_{r}$ are the insertion losses of one optical link and one optical router respectively, and $I\!L_{eo}$ and $I\!L_{oe}$ are the insertion losses of E-O and O-E converters respectively~\cite{phoenixsim}. Since $I\!L_{l}$, $I\!L_{r}$, $I\!L_{eo}$, and $I\!L_{oe}$ are constant parameters for specific optical devices and router structure, the worst-case insertion loss $I\!L_{wc}$ is decided by the maximum length of the routing paths.  Similarly, in terms of signal crosstalk, longer routing path passing through more optical elements can also lead to more crosstalk during transmission. Table \ref{tab:crosstalk} shows the maximum length of routing paths for the three mapping strategies during the FCNN training, indicating that FM has the least insertion loss and crosstalk ranked as (1), and RRM has the most insertion loss and crosstalk ranked as (3):

  
\begin{table}[!ht]
	\vspace{-1mm}
	\captionsetup{justification=centering}
	\caption{Maximum path length for Three Mapping Methods}
	\label{tab:crosstalk}
	\centering
	\def\arraystretch{2}
	\setlength{\tabcolsep}{1mm}
	\begin{tabular}{ c  c  c}
		\hline 
		\! Name of mapping \! & \! Maximum path length \! & \! Rank \\
		\hline 
		\! FM \! & \! ${\max}^{l}_{i=1}(m_i-1)$\! & \! 1\\
		
		\! ORRM \!& \!${\max}^{l}_{i=2}(m_i+m_{i-1}-r_i)$\! & \! 2\\
		
		\! RRM \! & \!${\max}^{l}_{i=2}(m_i+m_{i-1}-1)$\! & \! 3\\
		
		\hline
	\end{tabular}
\end{table} 

\vspace{-2mm}
\subsection{Memory Analysis}

We use $\psi$ to denote the memory requirement for storing one parameter such as weight, bias, gradient, input and output in distributed SRAM of cores and $\mu$ to represent the batch size number in one epoch ($\mu\geq1$). As neurons are involved in both FP and BP processes, we calculate the memory requirement for one neuron. For the FP process, according to Eq. (\ref{eq:f}), the memory requirement to process one neuron in layer $i$ includes $n_{i-1}$ weights, one bias, $n_{i-1}$ inputs, and one output, with a total memory requirement of $2(n_{i-1}+1)\mu \psi$. For the BP process, since the weight and bias have already been stored in the SRAM of the corresponding core, the additional parameters to be stored include one bias gradient, $n_{i-1}$ weight gradients, one learning rate, with a total extra memory requirement of $(n_{i-1}+2)\mu \psi$. Hence, when batch size is equal or lager than 2, the total memory requirement of each neuron on layer $i$ in both FP and BP process can be calculated as $s_i=(3n_{i-1}+4)\mu \psi, i \in [1,l]$.
The memory requirement for SRAM in each core depends on the number of mapped neurons and the mapping strategy, which can be calculated as $\sum^{l}_{i=1}\sum^{n_i}_{j=1}M[i,j,k] \times s_i$. So, the maximum memory requirement among SRAMs of all the cores (worst case) for the mapping strategy is formulated as:
\begin{equation}\label{eq:mapm}
 \max^{m}_{k=1}\sum^{l}_{i=1}\sum^{n_i}_{j=1}M[i,j,k]\times s_i.
\end{equation}
Therefore, we estimate the maximum memory requirement among cores for three mapping strategies by Eq. (\ref{eq:mapm}) ranking from low (1) to high (3), as shown in Table \ref{tab:memory}:
\begin{table}[!ht]
	\bgroup
	\vspace{-1mm}
	\captionsetup{justification=centering}
	\caption{Memory Requirements for Three Mapping Methods}
	\label{tab:memory}
	\centering
	\def\arraystretch{2}
	\setlength{\tabcolsep}{1mm}{
		\begin{tabular}{ c  c  c}
			\hline
			\! Mapping method \! & \! Maximum memory requirement of core \! & \! Rank\\
			\hline
			
			\! RRM \! & \!$\max_{i=1}^{l}\left[\frac{(3n_{i-1}+4) \mu \psi  n_i}{m_{i}^{*}}\right]$\! & \! 1\\
			
			\! ORRM \!& \!$\max^{l-1}_{i=1} \left[\frac{(3n_{i-1}+4) n_i}{m_{i}^{*}} + \frac{(3n_{i}+4) n_{i+1}}{m_{i+1}^{*}}\right]\mu \psi $\! & \! 2\\
			
			\! FM \! & \!$\sum^{l}_{i=1}\frac{ n_i \mu \psi {(3n_{i-1}+4)}}{m^{*}_i}$\! & \! 3\\
			
			\hline
	\end{tabular}}
	\egroup
\end{table}

The memory requirement in TABLE \ref{tab:memory} is calculated under the condition that periods are within one round of the ring (i.e. $\sum_{i=1}^l m^*_i \leq m$ for RRM or $\sum_{i=1}^l (m^*_i-r_i) \leq m$ for ORRM). When periods are covering more than one round of the ring (i.e. $\sum_{i=1}^l m^*_i> m$ for RRM or $\sum_{i=1}^l (m^*_i-r_i)> m$ for ORRM), the calculation needs to add more items according to the number of rounds. Note that the cores only load instructions and data from main memory in the initialization process according to the mapping strategy. When the memory requirement calculated above does not exceed the SRAM capacity of cores, there is no main memory access during the neural network training. Otherwise, additional main memory accesses are required causing extra delay for the training time.

\vspace{-1mm}
\subsection{Routing and Wavelength Assignment}\label{sec:routing}
\noindent We describe the details of routing and wavelength assignment, and then use an example to illustrate the control process in the ring-based ONoC by using RRM. We assume the communication is bidirectional, where transmission direction in FP is clockwise and transmission direction in BP is anticlockwise. 
\begin{figure}[!ht]
	\vspace{-3mm}
	\centering
	\includegraphics[scale=0.5]{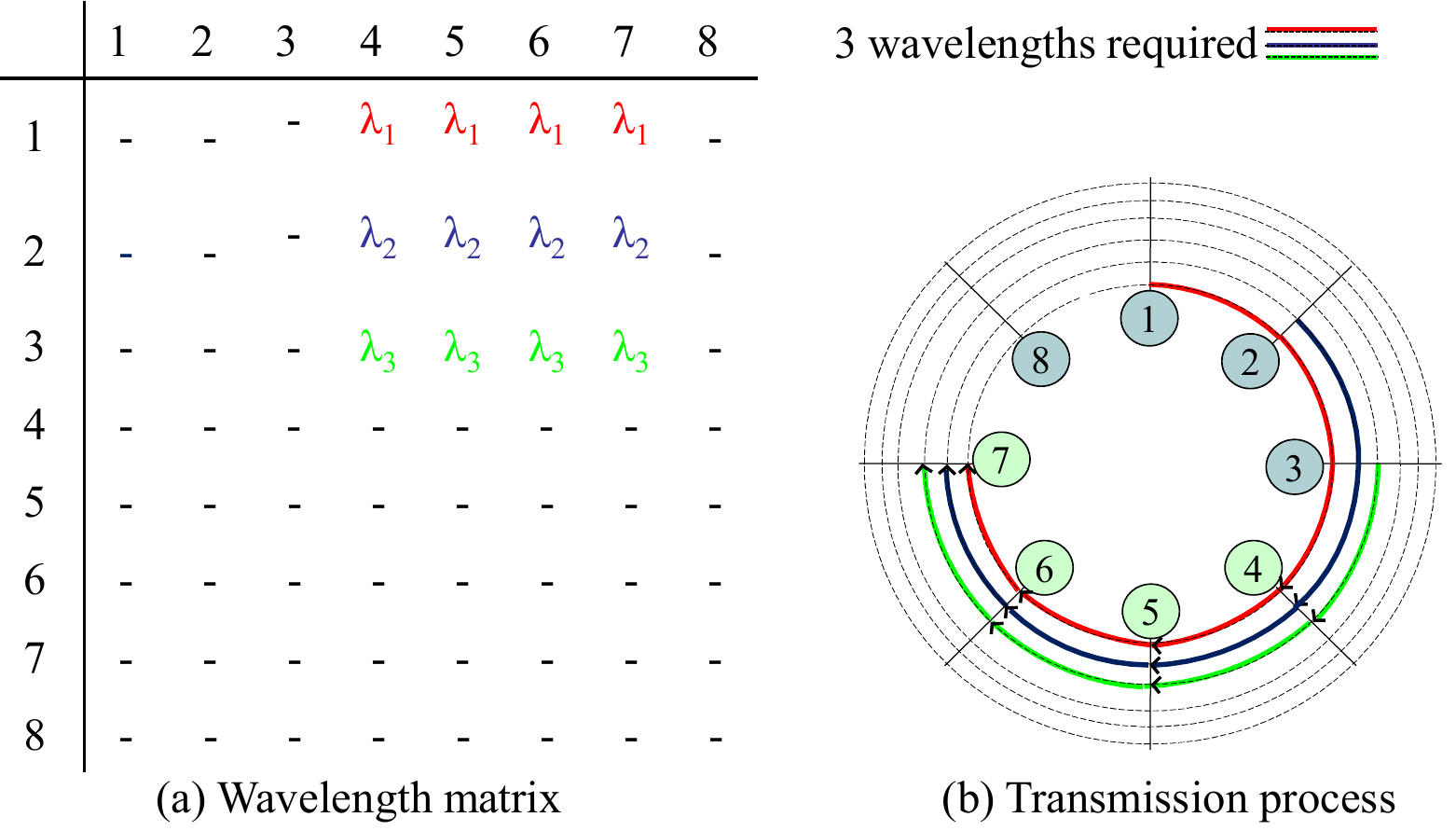}
	{
		\vspace{-1mm}
		\caption{Routing and wavelength assignment example from Period 1 to 2 in FP : (a) Wavelength matrix generated by RWA and (b) Communications process.}
		\label{fig:transmission}
	}
\end{figure}
\vspace{-3mm}

As mentioned in Section 2.2, the ONoC architecture consists of electrical core, optical control and optical data planes. The processes of routing and wavelength control are managed by the optical control plane. 
Firstly, the manager core calculates the optical number of cores and sends request to Routing and Wavelength Allocator (RWA), and then RWA generates the corresponding wavelength matrix and optical control packets. Secondly, RWA broadcasts the control packets along the cyclic optical control channel to configure the optical routers based on the generated wavelengths matrix. After control packets are received by the optical routers, the modulators and drop filters of the corresponding optical routers are configured with routing paths setup from each core in the previous period to all the cores in the next period. A specific wavelength for each routing path is specified according to the wavelength matrix, and multiple routing paths can transmit in parallel by different wavelengths through the same waveguide by WDM. 

As shown in Fig.~\ref{fig:transmission} (a), the wavelength matrix, denoted by $WM$, determines the routing and wavelength allocation, where $WM[i,j]=\lambda_k$ indicates \textit{core i} should communicate with \textit{core j} using wavelength $\lambda_k$. Assuming that 3 cores $[core_1, core_2, core_3]$ are assigned for \textit{Period 1} and 4 cores $[core_4,core_5, core_6, core_7]$ are assigned for \textit{Period 2} using RRM strategy. According to the $WM$, $\lambda_1$ is assigned for the optical path from the source $[core_1]$ to the destinations $[core_4,core_5, core_6, core_7]$, with $\lambda_2$ and $\lambda_3$ assigned to optical paths from the sources $[core_2]$ and $[core_3]$ respectively. With all sender's modulators and receivers' drop filters activated, these three optical paths from source cores $[core_1, core_2, core_3]$ allocated for \textit{Period 1} can send their data to cores $[core_4,core_5, core_6, core_7]$ allocated for \textit{Period 2} simultaneously by three different wavelengths $\lambda_1$, $\lambda_2$ and $\lambda_3$, as illustrated in Fig. \ref{fig:transmission} (b).

\vspace{-2mm}
\section{Evaluation} \label{sec:evaluation}

In this Section, we conduct extensive simulations to verify the effectiveness of our methods. We first introduce our simulation setup in Section 5.1, and present the comparisons between theoretical results and simulation results in Section 5.2. Then, we carry out simulations to compare our methods with exiting methods in Section 5.3, and compare our methods on ONoCs with ENoCs in Section 5.4.
\vspace{-1mm}
\subsection{Simulation settings}
\noindent The simulation platform is implemented in a machine with intel i5 3200 CPU and 32 Gb main memory. Gem5~\cite{binkert2011gem5} is used to simulate the target ONoC and ENoC systems, and DSENT~\cite{sun2012dsent} is used in calculating the energy consumption. To collect computation time and communication traces, we implement the FCNN in C using GNU Scientific Library and BLAS gemm~\cite{kaagstrom1998gemm}. Sigmoid is used as the non-linear activation function in hidden layers and softmax is used as the 
activation function in the output layer. We run the configured workloads with up to 1000 threads to generate the communication traces for up to 1000 cores. The communication traces are fed into Gem5 to obtain the communication time of the simulated ONoC and ENoC systems. Based on the simulated results of Gem5, we calculate the energy consumption of ONoC and ENoC using the energy model in \cite{grani2014design}, where the ONoC parameters include laser power, optical insertion loss, MR thermal tuning power, and etc. The values of these ONoC parameters are retrieved from DSENT. To get accurate computation time of each core, we repeat the computation workload of each core for a thousand times and then obtain the average time. In this way, we make sure the computation is carried out in the CPU caches, which matches our ONoC architecture.

The parameters of the simulated cores are shown in Table~\ref{tab:architecture}, and the parameters of the simulated ONoC are shown in Table~\ref{tab:parameters}, which are obtained from~\cite{nicolescu2017photonic}~\cite{van2018effect}~\cite{MPNoC}~\cite{vlasov2008high}. Note that the size of distributed SRAM in Table 4 is the maximum memory requirement of our methods calculated by the worst case in Fixed mapping using NN benchmarks under batch size 128. 
The bandwidth shown in Table~\ref{tab:parameters} is the bandwidth per wavelength. We set $\phi=1$ (refer to Eq. (9)) with the assumption that all cores can be utilized without limitation, while in practical $\phi$ can be adjusted according to the system limitation. The FCNN models used in our simulation are listed in Table~\ref{tab:NN}, which are well-known models for processing fashion-mnist~\cite{cirecsan2010deep} and cifar-10 dataset~\cite{lin2015far}, with high classification accuracy. The number of neurons in the output layers is 10, and the number of neurons in the input layers is 784 or 1024. The number of neurons in the hidden layers varies between 500 and 4000.


\begin{table}[!ht]
	\captionsetup{justification=centering}
	\caption{Parameters of core architecture and memory hierarchy}
	\vspace{-1mm}
	\label{tab:architecture}
	\centering
	\setlength{\tabcolsep}{1mm}{
		\begin{tabular}{c  c }
			\hline
			\!Parameter\! & \!Value\\ 
			\hline
			\! Core frequency \! & \!3.4 \textit{GHz}\\ 
			\! Core Rmax \! & \!6 \textit{GFLOPS}\\ 
			\! Private L1 (I cache/ D cache) \! & \!128/128 \textit{KB}\\ 
			\! L1 latency\! & \!1 \textit{cycle}\\ 
			\! Distributed SRAM \! & \!82.5 \textit{M}\\ 
			\! Distributed SRAM latency\! & \!~~~~10 \textit{cycles} (front end/back end)\\
			\! Memory controller latency\! & \!6 \textit{cycles}\\ 
			\! Bandwith of main memory\! & \!10 \textit{Gb/s}\\
			\hline
	\end{tabular}}
	\vspace{-2mm}
\end{table}

\begin{table}[!ht]
	\captionsetup{justification=centering}
	\caption{ONoC Parameters}
	\vspace{-1mm}
	\label{tab:parameters}
	\centering
	\setlength{\tabcolsep}{1mm}{
		\begin{tabular}{c  c  c  c }
			\hline
			\!Parameter\! & \!Value\! & \!Parameter\!& \!Value\\ 
			\hline
			\! Bandwidth \! & \!40 \textit{Gb/s}\! & \!Modulation speed \! & \!10 \textit{Gb/s}\\ 
			\! OE/EO delay \! & \! 1 \textit{cycle/flit} & \!Wavelength number \! & \!8/64\\ 
			\!Time of flight\! & \! 1 \textit{cycle/flit} & \!Waveguide propagation\! & \!1.5 \textit{dB/cm}\\ 
			\! Serialization delay\!  & \!2 \textit{cycles/flit}\! & \!Waveguide crossing\! & \! 1 \textit{dB}\\ 
			\!Splitter\! & \!0.5 \textit{dB}\! & \!Waveguide bending\! & \!0.005 \textit{dB/90$^{o}$}\\ 
			\!MR pass\! & \!0.005 \textit{dB/MR}\! & \!Laser efficiency\! & \!30\%\!\\
			\!MR drop\! & \!0.5 \textit{dB/MR}\! & \!Coupler \! & \!1 \textit{dB}\!\\
			\hline
	\end{tabular}}
	\vspace{-2mm}
\end{table}

\begin{table}[!ht]
	\vspace{-1mm}
	\captionsetup{justification=centering}
	\caption{Neural network list }
	\label{tab:NN}
	\centering
	\begin{tabular}{ l l }
		\hline
		NN1 & 784$\textendash$1000$\textendash$500$\textendash$10 \\
		\hline
		NN2 & 784$\textendash$1500$\textendash$784$\textendash$1000$\textendash$500$\textendash$10 \\
		\hline
		NN3 & 784$\textendash$2000$\textendash$1500$\textendash$784$\textendash$1000$\textendash$500$\textendash$10 \\
		\hline
		NN4 &784$\textendash$2500$\textendash$2000$\textendash$1500$\textendash$784$\textendash$1000$\textendash$500$\textendash$10  \\
		\hline
		NN5 & 1024$\textendash$4000$\textendash$1000$\textendash$4000$\textendash$10 \\
		\hline
		NN6 & 1024$\textendash$4000$\textendash$1000$\textendash$4000$\textendash$1000$\textendash$4000$\textendash$1000$\textendash$4000$\textendash$10 \\
		\hline
		\vspace{-3mm}
	\end{tabular}
\end{table}

\begin{table}[!htbp]
		\centering
		\captionsetup{justification=centering}
		\vspace{-1mm}
		\caption{Prediction accuracy for optimal number of cores}
		\label{tab:optimal_solution}
	\begin{tabular}{c c c} 
		\hline
		Neural network & APE (\%) & APD (\%) \\ \hline
		NN1            & 0.83         & 1.26         \\  
		NN2            & 1.24         &  2.78        \\ 
		NN3            & 1.82         &  2.53        \\ 
		NN4            & 1.92        &  3.68        \\ 
		NN5            & 0.93        &   1.42       \\ 
		NN6            & 2.22        &  4.29       \\ \hline
	\end{tabular}
\end{table}

\begin{figure*}[!htbp]
	\centering
	\vspace{-2mm}
	\centering
	\includegraphics[width=\textwidth]{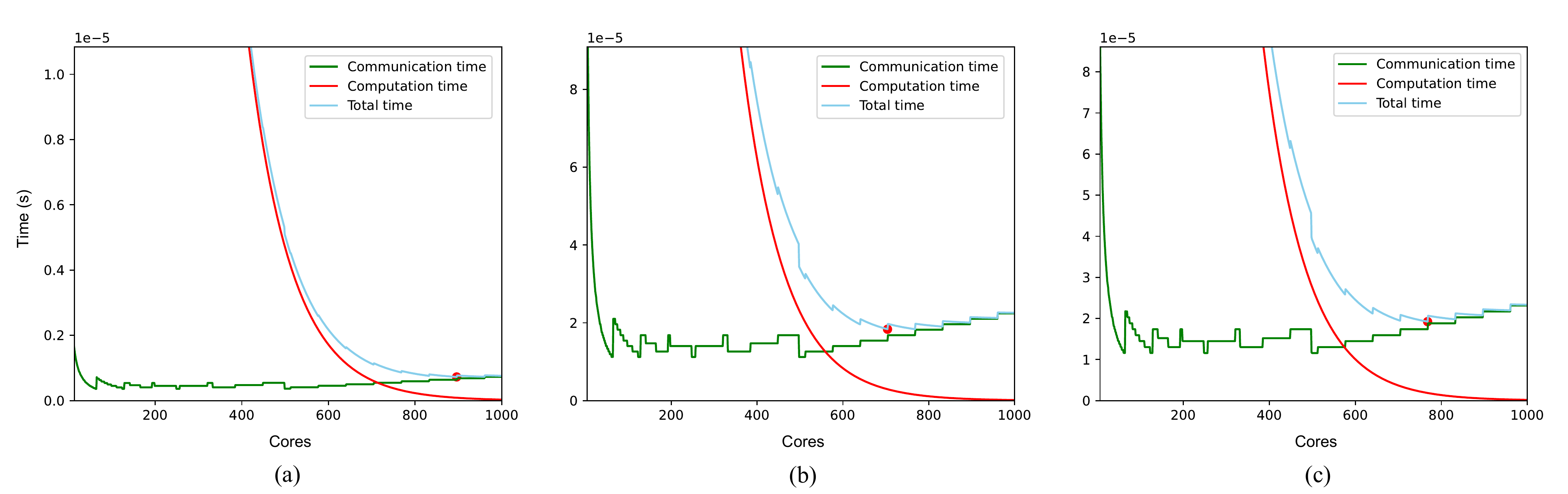}
	\vspace{-4mm}
	\caption{Performance of NN2 for (a) Period 3 (layer 3 during FP); (b) Period 8 (layer 3 during BP); (c) combined results for both Period 3
and 8 (layer 3 during both FP and BP).}
	\label{fig:nn4}
	\vspace{-3mm}
\end{figure*}

\vspace{-1mm}
\subsection{Optimal number of cores}
\noindent  In this experiment, we simulate the training process of the NN benchmarks for each layer under different numbers of cores, ranging from 1 to 1000, in order to find the optimal number of cores for each layer of each benchmark in simulated environment. We use Batch Sizes (BS) 1, 8, 32, 64 and wavelength numbers 8 and 64 in the simulation. 

To show Lemma 1 can predict the optimal number of cores accurately, first we calculate the Average Prediction Error (APE) and Average Performance Difference (APD) for the optimal number of cores for each layer of each NN benchmark. The values in Table \ref{tab:optimal_solution} are the average value calculated by using the simulated and theoretical optimal number of cores (lemma 1) generated under different batch sizes and wavelength numbers. As shown in Table \ref{tab:optimal_solution}, the APE of NN benchmarks is within 2.3\% and the APD is within 5\%, which verifies the prediction accuracy of our model.  

We use an example to illustrate how we calculate the optimal number of cores in one layer by simulation. Fig. \ref{fig:nn4} shows the performance of FP and BP in layer 3 of NN2, with batch size set to 32 and the number of wavelengths set to 64.
The x-axis of the figure is the number of cores ranging from 1 to 1000.
The three curves in each subgraph of Fig. \ref{fig:nn4} represent the computation time, communication time and the total time of the layer, where Fig. \ref{fig:nn4} (a) is for layer 3 during FP, Fig. \ref{fig:nn4} (b) is for layer 3 during BP, and Fig. \ref{fig:nn4} (c) is the combined results for layer 3 during both FP and BP. It can be seen from Fig. \ref{fig:nn4} (a) and (b) that the optimal number of cores marked by red dots for FP and BP are 896 and 704 respectively, because computations and communications workload in FP and BP are different. Since the same cores need to be assigned to the same layer to guarantee the data locality, we calculate the optimal number of cores by using the combined results in both FP and BP processes. As shown in Fig. \ref{fig:nn4} (c), the optimal number of cores assigned for layer 3 during both BP and FP is 769.

\begin{table}[!htbp]
	\captionsetup{justification=centering}
	\caption{Performance improvement of the optimal solution over FNP and FGP.}
	\vspace{-1mm}
	\label{tab:performance}
	\renewcommand\arraystretch{1}
	\begin{tabular}{|c|c|c|c|c|c|c|}
		\hline
		\multicolumn{2}{|c|}{} & BS 1 & BS 8 & BS 64 & BS 128 & Average \\ \hline
		\multirow{2}{*}{NN1}         &FNP        & 10.66\%      & 9.85\%       & 9.86\%        & 9.96\%         & 10.08\% \\ \cline{2-7} 
		&FGP        & 9.95\%       & 1.28\%       & 0.00\%        & 0.00\%         & 2.81\%  \\ \hline
		\multirow{2}{*}{NN2}         &FNP        & 13.52\%      & 18.02\%      & 20.53\%       & 20.84\%        & 18.23\% \\ \cline{2-7} 
		&FGP        & 19.32\%      & 3.51\%       & 0.26\%        & 0.10\%         & 5.80\%  \\ \hline
		\multirow{2}{*}{NN3}         &FNP        & 19.28\%      & 24.96\%      & 27.61\%       & 27.94\%        & 24.95\% \\ \cline{2-7} 
		&FGP        & 16.87\%      & 3.45\%       & 0.19\%        & 0.04\%         & 5.14\%  \\ \hline
		\multirow{2}{*}{NN4}         &FNP        & 21.49\%      & 29.25\%      & 31.93\%       & 32.36\%        & 28.76\% \\ \cline{2-7} 
		&FGP        & 17.75\%      & 4.56\%       & 0.17\%        & 0.06\%         & 5.63\%  \\ \hline
		\multirow{2}{*}{NN5}         &FNP        & 17.00\%      & 26.63\%      & 31.16\%       & 31.62\%        & 26.60\% \\ \cline{2-7} 
		&FGP        & 12.76\%      & 4.22\%       & 0.64\%        & 0.29\%         & 4.48\%  \\ \hline
		\multirow{2}{*}{NN6}         &FNP        & 15.66\%      & 24.68\%      & 29.64\%       & 30.24\%        & 25.05\% \\ \cline{2-7} 
		&FGP        & 14.98\%      & 6.09\%       & 0.93\%        & 0.42\%         & 5.61\%  \\ \hline
	\end{tabular}
\end{table}

\begin{table}[!htbp]
	\captionsetup{justification=centering}
	\caption{Energy difference of the optimal solution over FNP and FGP.}
	\vspace{-1mm}
	\label{tab:energy}
	\renewcommand\arraystretch{1}
	\renewcommand\tabcolsep{5.3pt}
	\begin{tabular}{|c|c|c|c|c|c|c|}
		\hline
		\multicolumn{2}{|l|}{} & BS 1     & BS 8     & BS 64    & BS 128   & Average  \\ \hline
		\multirow{2}{*}{NN1}  & FNP  & -1.34\%  & -3.19\%  & -5.69\%  & -5.57\%  & -3.95\%  \\ \cline{2-7} 
		& FGP  & 12.70\%  & 3.44\%   & 0.00\%   & 0.00\%   & 4.04\%   \\ \hline
		\multirow{2}{*}{NN2}  & FNP  & -4.34\%  & -1.62\%  & -5.34\%  & -4.91\%  & -4.05\%  \\ \cline{2-7} 
		& FGP  & 27.80\%  & 11.48\%  & 2.71\%   & 2.56\%   & 11.14\%  \\ \hline
		\multirow{2}{*}{NN3}  & FNP  & -3.31\%  & 2.59\%   & -0.74\%  & -2.02\%  & -0.87\%  \\ \cline{2-7} 
		& FGP  & 26.82\%  & 13.56\%  & 4.54\%   & 2.94\%   & 11.96\%  \\ \hline
		\multirow{2}{*}{NN4}  & FNP  & -5.58\%  & 3.07\%   & 0.00\%   & 2.50\%   & 0.00\%   \\ \cline{2-7} 
		& FGP  & 28.69\%  & 15.14\%  & 5.19\%   & 2.46\%   & 12.87\%  \\ \hline
		\multirow{2}{*}{NN5}  & FNP  & -5.24\%  & 4.05\%   & 7.95\%   & 8.13\%   & 3.72\%   \\ \cline{2-7} 
		& FGP  & 18.70\%  & 8.13\%   & 2.51\%   & 1.59\%   & 7.73\%   \\ \hline
		\multirow{2}{*}{NN6}  & FNP  & -26.01\% & -15.86\% & -17.47\% & -17.72\% & -19.27\% \\ \cline{2-7} 
		& FGP  & 27.31\%  & 15.82\%  & 4.97\%   & 3.11\%   & 12.80\%  \\ \hline
	\end{tabular}
\end{table}

\begin{table*}[htbp]
	\captionsetup{justification=centering}
	\caption{Optimal number of cores}
	\vspace{-1mm}
	\label{tab:optimal_number}
	\centering
	\setlength{\tabcolsep}{2mm}{
		\begin{tabular}{lllll}
			\hline
 
			\multicolumn{1}{|l|}{}                  & \multicolumn{1}{l|}{BS = 1, wavelength = 8}                 & \multicolumn{1}{l|}{BS = 1, wavelength = 64}                       & \multicolumn{1}{l|}{BS = 8, wavelength = 8}                                                                    & \multicolumn{1}{l|}{BS = 8, wavelength = 64}                                                                   \\ \hline
			\multicolumn{1}{|l|}{NN1}               & \multicolumn{1}{l|}{[1000, 257, 10]}                     & \multicolumn{1}{l|}{[1000, 257, 10]}                            & \multicolumn{1}{l|}{[1000, 257, 10]}                                                                       & \multicolumn{1}{l|}{[1000, 500, 10]}                                                                      \\ \hline
			\multicolumn{1}{|l|}{NN2}               & \multicolumn{1}{l|}{\begin{tabular}[c]{@{}l@{}}[1000, 393, 505, 257, 10] \end{tabular}}             & \multicolumn{1}{l|}{\begin{tabular}[c]{@{}l@{}}[1000, 449, 513, 321, 10] \end{tabular} }                    & 
			\multicolumn{1}{l|}{\begin{tabular}[c]{@{}l@{}}[1000, 401, 505, 329, 10] \end{tabular} }                                                               & \multicolumn{1}{l|}{\begin{tabular}[c]{@{}l@{}}[1000, 784, 641, 500, 10] \end{tabular}}                                                              \\ \hline
			\multicolumn{1}{|l|}{NN3}               & \multicolumn{1}{l|}{\begin{tabular}[c]{@{}l@{}}[1000, 753, 393, 505, 257, 10] \end{tabular} }         & \multicolumn{1}{l|}{\begin{tabular}[c]{@{}l@{}}[1000, 769, 449, 513, 257, 10]] \end{tabular} }                
			& \multicolumn{1}{l|}{\begin{tabular}[c]{@{}l@{}}[1000, 753, 393, 505, 257, 10]\end{tabular} }                                                            & \multicolumn{1}{l|}{\begin{tabular}[c]{@{}l@{}}[1000, 1000, 784, 513, 257, 10] \end{tabular} }                                                         \\ \hline
			\multicolumn{1}{|l|}{NN4}               & \multicolumn{1}{l|}{ \begin{tabular}[c]{@{}l@{}}[1000, 673, 753, 393, 505,\\ 257, 10] \end{tabular}  }     
			& \multicolumn{1}{l|}{\begin{tabular}[c]{@{}l@{}}[1000, 705, 769, 449, 577,\\ 257, 10] \end{tabular} }           
			& \multicolumn{1}{l|}{\begin{tabular}[c]{@{}l@{}}[1000, 681, 753, 417, 553,\\ 257, 10] \end{tabular} }                                                       & \multicolumn{1}{l|}{\begin{tabular}[c]{@{}l@{}}[1000, 897, 1000, 449, 577,\\ 500, 10] \end{tabular} }                                                     \\ \hline
			\multicolumn{1}{|l|}{NN5}               & \multicolumn{1}{l|}{ [1000, 505, 665, 10]}                
			& \multicolumn{1}{l|}{[1000, 833, 833, 10]}                      
			& \multicolumn{1}{l|}{[1000, 681, 817, 10]}                                                                & \multicolumn{1}{l|}{[1000, 1000, 833, 10]}                                                                \\ \hline
			\multicolumn{1}{|l|}{NN6}               & \multicolumn{1}{l|}{\begin{tabular}[c]{@{}l@{}}[1000, 513, 801, 513,  801,\\ 513, 801, 10]\end{tabular}} &  \multicolumn{1}{l|}{\begin{tabular}[c]{@{}l@{}}[1000, 513, 833, 513, 833,\\ 513, 833, 10] \end{tabular}} & \multicolumn{1}{l|}{\begin{tabular}[c]{@{}l@{}} [1000, 513, 801, 513, 801,\\ 513, 801, 10]\end{tabular}} & \multicolumn{1}{l|}{\begin{tabular}[c]{@{}l@{}}[1000, 1000, 833, 1000, 833,\\ 1000, 833, 10] \end{tabular}} \\ \hline
			&                                                     &                                                            &                                                                               &                                                                                                      
	\end{tabular}}
\end{table*}

\begin{figure*}[hbpt]
	\centering
	\includegraphics[width=\textwidth]{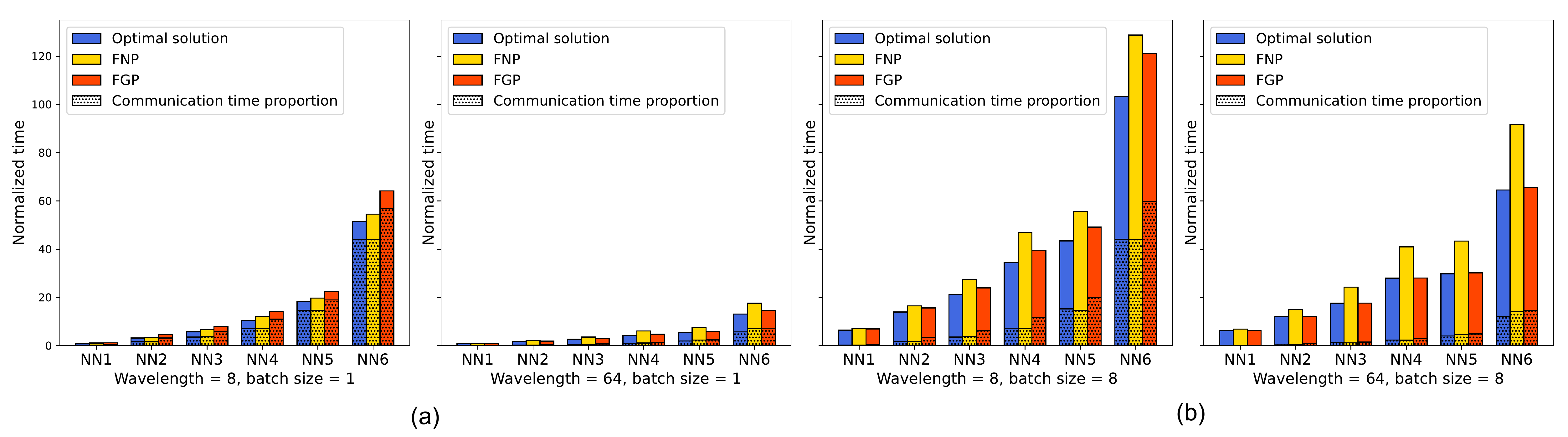}
	\vspace{-4mm}
	\caption{Comparisons on the training time by three methods for NN benchmarks with 8 and 64 wavelengths under (a) batch size 1 and (b) batch size 8. The shaded part quantifies the communication time among the total training time.}
	\label{fig:traditional}
	\vspace{-4mm}
\end{figure*}

\begin{figure*}[htbp]
	\centering
	\includegraphics[width=\textwidth]{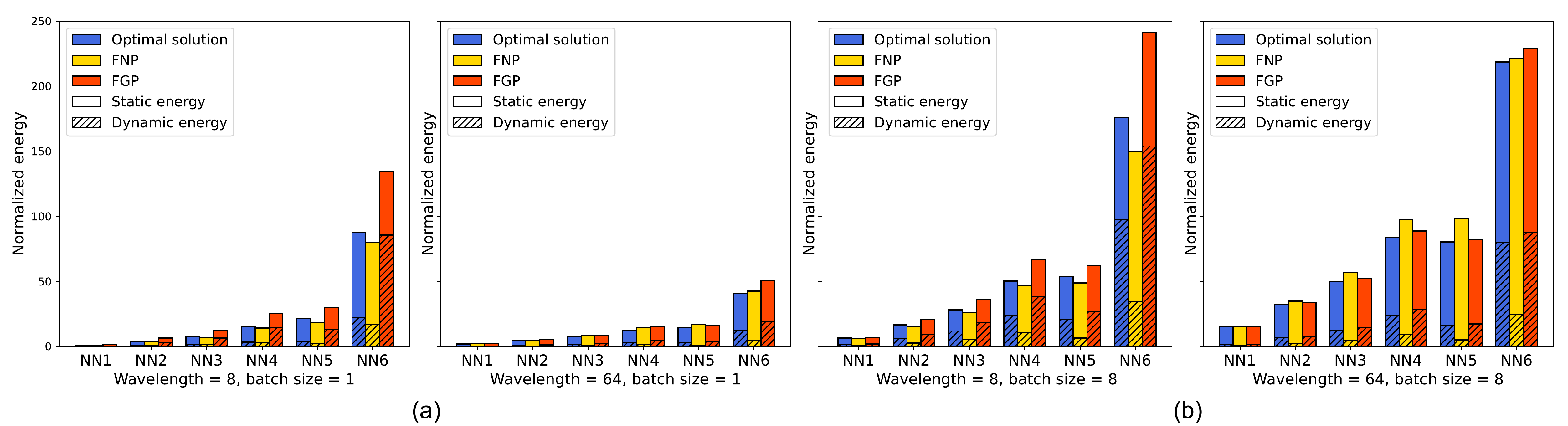}
	\vspace{-4mm}
	\caption{Comparisons on the normalized energy by three methods for NN benchmarks with 8 and 64 wavelengths under (a) batch size 1 and (b) batch size 8. The shaded part represents for dynamic energy consumption while the un-shaded part represents for static energy consumption.}
	\label{fig:O_energy}
	\vspace{-4mm}
\end{figure*}

\vspace{-1mm} 
\subsection{Comparisons with Traditional Methods}
To show better performance and energy efficiency of our proposed optimal solution, two traditional methods, (1) Finest-Grained Parallel method (FGP)~\cite{pethick2003parallelization} by mapping one neuron to one core; (2) Fixed Number Parallel method (FNP)~\cite{zou2019learn}, which is set with a fixed number of cores (200) for each period during the training, are compared with the optimal solution. When comparing the above three methods, we adopt Fixed Mapping strategy in the simulation, and use Batch Sizes (BS) 1, 8, 64, 128 and wavelength numbers 8 and 64 in the simulation.

Table \ref{tab:performance} shows the average performance improvement for the optimal solution over FNP and FGP for each NN benchmark using wavelength numbers 8 and 64. As can be seen from Table \ref{tab:performance}, with the increasing of batch sizes, the performance improvement for the optimal solution compared with FNP is increasing, while performance improvement compared with FGP is decreasing for all NN benchmarks. That is because computation workload is increasing with the increase of batch size, thus the optimal solution tends to use more cores getting close to FGP.
Table \ref{tab:energy} shows the energy difference for the optimal solution with FNP and FGP for each NN benchmark. It can be seen that the optimal solution is more energy efficient than FGP, and less energy efficient than FNP in most cases. 

Fig. \ref{fig:traditional} and Fig. \ref{fig:O_energy} show the comparisons on performance and energy consumption between the optimal solution and two traditional methods with NN benchmarks under batch sizes 1 and 8 with wavelength numbers 8 and 64. All results are normalized by dividing the first result of NN1. The corresponding optimal number of cores generated by our optimal solution is shown in Table \ref{tab:optimal_number}. From Fig. \ref{fig:traditional} (a) and (b), we can see that the total training time of the optimal solution is the lowest among three methods under different batch sizes and wavelengths. On average, when the batch size is set to 1 and the number of wavelengths is set to 8 or 64, the training time of the optimal solution compared with FGP and FNP is reduced by 16.27\% and 15.27\% respectively. When the batch size is 8, the training time of the optimal solution is reduced by 22.23\% and 3.85\%, respectively. It can be seen from Fig. \ref{fig:traditional} (a) and (b) that the training time of FCNN using 64 wavelengths is less than that using 8 wavelengths. This is because the communication time is largely reduced due to concurrent communications achieved by WDM with more wavelengths and less time slots.

Fig. \ref{fig:O_energy} (a) and (b) show that the optimal solution consumes the least energy among three methods under wavelength number 64, and FNP consumes the least energy under wavelength number 8 for both batch sizes 1 and 8. That is because the static energy affected by the training time is dominated in the total energy consumption under wavelength number 64, while dynamic energy affected by the number of cores takes up more proportion in total energy consumption under wavelength number 8. This indicates that our method tends to be more energy-efficient especially when the number of wavelengths is large (e.g. 64 in the example). On average, when the batch size is set to 1 and the number of wavelengths is set to 8 or 64, the energy consumption of the optimal solution compared with FGP is reduced by 23.67\% but is increased by 7.64\% compared with FNP. 
When the batch size is 8, the energy consumption of the optimal solution compared with FGP is reduced by 11.26\% but is increased by 1.83\% compared with FNP.

\begin{figure*}[htbp]
	\centering
	\includegraphics[scale=0.36]{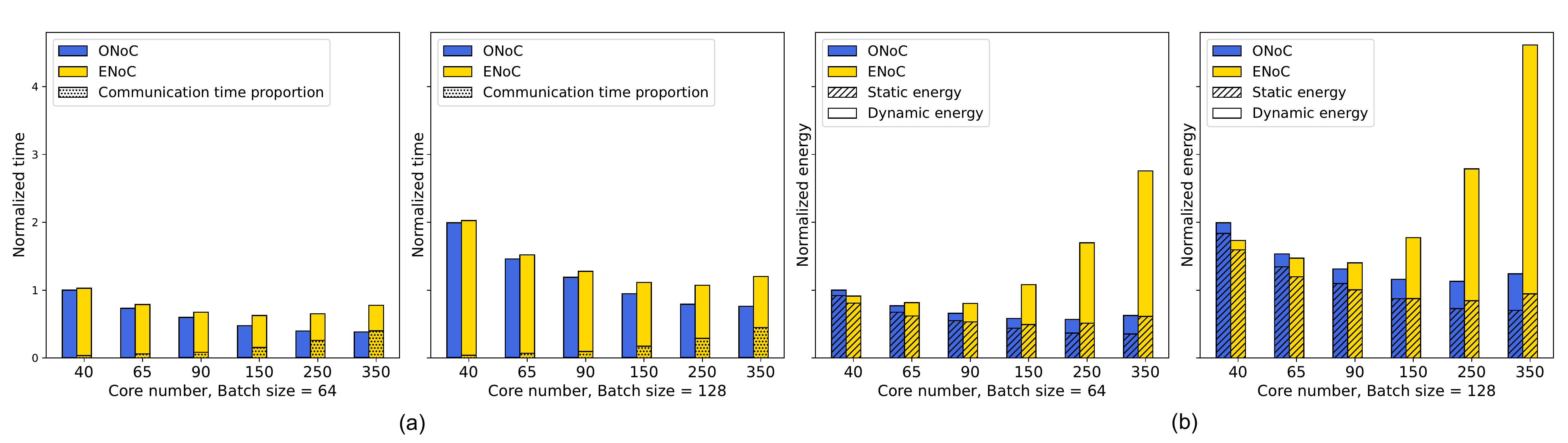}
	\vspace{-4mm}
	\caption{Comparisons between ONoC and ENoC using NN2 under batch size 64 and 128 for (a) performance and (b) energy consumption.}
	\label{fig:NoC}
	\vspace{-4mm}
\end{figure*}

\vspace{-1mm} 
\subsection{Comparisons with ENoC}

In order to demonstrate the advantages of using ONoCs, we compare our methods on ONoC with ENoC in terms of both performance and energy consumption.
The performance of the three mapping strategies (FM, RRM, and ORRM) in ONoC are almost the same because latency is not affected much by the transmission distance in ONoC due to the very high transmission speed of optical signals. However, different mapping strategies will make a huge difference in training time for ENoC because of the hop by hop routing in ENoC. In order to make a fair comparison, we compare the performance and energy consumption between ONoC with ENoC by Fixed Mapping with a range of fixed number of cores (40, 65, 90, 150, 250, 350), rather than comparing based on the optimal number of cores. We use NN2 with batch sizes 64 and 128 because NN2 has a moderate size, neither too small nor too deep. The parameters for ONoC and ENoC in Gem5 are set as follows. The packet size for both ONoC and ENoC is set to 64~\textit{bytes} and 16 bytes/flit. The number of wavelengths used for ONoC is 64. The latency of hop by hop routing in ENoC is set to 2 cycles. Shortest path routing is used for the 4-channel electric routers in ENoC.

Fig. \ref{fig:NoC} (a) and (b) show the comparisons of performance and energy consumption between ONoC and ENoC for NN2 with batch size 64 and 128 respectively. It can be seen from Fig. \ref{fig:NoC} (a) that FCNN training time on ONoC is smaller than ENoC for both batch sizes. Compared with ENoC, the average time reductions for ONoC under batch size 64 and 128 are 21.02\% and 12.95\% respectively. That is because ONoC allows concurrent communications using WDM to deal with large volume of core-to-core communications resulting in less communication time. When the number of cores is increasing, the performance difference is becoming more obvious for both batch sizes, because ENoC has larger communication latency among cores due to increasing routing delays. 

From Fig. \ref{fig:NoC} (b), we can see that the energy consumption of ONoC is larger than ENoC when the number of cores is small (no more than 90 in the example), and much smaller than ENoC when the number of cores is large (more than 90 in the example). This is because the static power is dominant in ONoC when the number of cores is small, while the dynamic energy in ENoC caused by the communication overhead is increasing rapidly when the number of cores is increased. 
On average, the energy consumption of ONoC is reduced by 47.85\% and 39.27\% compared with ENoC under batch size 64 and 128 respectively. 

In summary, the performance of our method on ONoC outperforms ENoC under different batch sizes, with more energy-efficiency especially when the number of cores is large.

\vspace{-2mm}
\section{Related Work} \label{sec:related}   
\noindent  There are many studies on parallel computing for neural network training with CPU, GPU, FPGA, ASIC, and etc.~\cite{ben2019demystifying}, which attempts to utilize the computing cores to accelerate neural networks. Due to the space limitation, we only show the related researches in the following aspects: 

(1) {\it Data Reuse Methods}: The state-of-the-art neural network acceleration approaches are mostly based on data reuse methods.  One data reuse method called {\it weight reuse} stashes the weight in the cache of cores during the training, which is used in co-processor design~\cite{chakradhar2010dynamically} and mobile co-processor designed for CNN training~\cite{gokhale2014240}. A second data reuse method is called {\it output reuse}, with the partial sum collected and maintained in the cache of cores, which is used in Shidiannao~\cite{du2015shidiannao}. Another data reuse method  called {\it row reuse} is proposed in Eyeriss~\cite{chen2016eyeriss}, where the computation of any given CNN shape is mapped onto the PE array. Our work reuses outputs and weights by storing the neural network parameters in different cores' SRAM distributively and reusing them in the back propagation.

(2) {\it Neural Network Acceleration Architectures}:
To reduce average data latency and energy consumption, a mesh NoC is proposed for FCNN training~\cite{liu2018neu}. 
To increase throughput and power efficiency, a 3D NoC architecture for ANN~\cite{firuzan2018reconfigurable}  and a NoC-based accelerator for CNN~\cite{chen2016eyeriss} are proposed. A microswitch NoC is proposed as a spatial neural network accelerator~\cite{kwon2017rethinking} with the consideration of latency, throughput, area and energy. More interestingly, a study on NoC based DNN accelerators~\cite{ascia2019networks} shows that the communication in neural networks accounts for the major delay in training and the memory accesses are accountable for most energy cost. 
All the above studies are based on ENoC, which has much higher latency and energy cost compared with ONoC.

\vspace{0 mm}
(3) {\it Mapping Methods}: 
In~\cite{multipleM}, several models for mapping neural networks on ENoC are proposed to get the high performance.
In~\cite{akopyan2015truenorth}, a mapping method is used by considering the wire-length between neurons for a TrueNorth neuromorphic system.
In~\cite{mand2012artificial}, mini-column used as basic functional element is mapped to each core, where each mini-column consists of 100 neurons. All of the above mapping strategies follow a sequential mapping scheme that maps neuron populations to cores one by one, which does not provide good performance for inter-neuron communications between different cores. Moreover, they cannot be applied to ONoC directly because they are not designed with the consideration of the ONoC's properties, thus cannot take advantage of optical transmissions. 
Besides, there are some studies for ONoC mapping problem by designing wavelength assignment~\cite{WAssignment}, or reducing crosstalk~\cite{crosstalk2} and thermal sensitivity influence~\cite{ThemalS} recently.
However, these approaches are all based on heuristic algorithm so they suffer from long training time to get the approximate optimal solution.

It is worth mentioning that there is {\it no research on the optimization of parallel FCNN training through the computation and communication trade-off in the context of ONoC}, which we develop to bridge the gap in this paper.

\vspace{-1 mm}
\section{Conclusions} \label{sec:conclusion}
\noindent In this paper, we propose a fine-grained parallel computing model to analyze the trade-off between communication and computation for training FCNN on ONoC with the objective of minimizing the total training time. Based on the model, we derive the optimal number of cores required for each training period in both forward propagation and back propagation processes. We further present three mapping strategies (FM, RRM and ORRM) and compare their advantages and disadvantages with regard to hotspot level, memory requirement and state transitions. Simulation results show that the average prediction error on the optimal number of cores using NN benchmarks is within 2.3\%. Extensive simulations demonstrate that our proposed methods can reduce FCNN training time by 22.28\% and 4.91\% on average compared with two traditional parallel computing methods, respectively. Compared with ENoC, our methods on ONoC can achieve 21.02\% and 12.95\% on reducing training time and 47.85\% and 39.27\% on saving energy under batch sizes 64 and 128, respectively.
Future work can be conducted for model extensions to other neural networks or other ONoC topologies.

\vspace{-1 mm}
\bibliographystyle{IEEEtran}
\bibliography{paper}

\vspace{-20mm}
\end{document}